\documentclass[acmsmall]{acmart}

\acmJournal{PACMPL}
\acmVolume{1}
\acmNumber{CONF} 
\acmArticle{1}
\acmYear{2018}
\acmMonth{1}
\acmDOI{} 
\startPage{1}

\setcopyright{none}

\usepackage{booktabs}
\usepackage{subcaption}
\usepackage{algorithm}
\usepackage[noend]{algpseudocode}
\usepackage{amsmath}
\usepackage{amssymb}
\usepackage{stmaryrd}
\usepackage{url}
\usepackage{listings}
\usepackage{mdwlist}
\usepackage{pifont}
\usepackage{graphicx}
\usepackage{wrapfig}
\usepackage{fancyvrb}
\usepackage{paralist}
\usepackage{caption}
\usepackage{comment}
\usepackage{bbm, dsfont}
\usepackage{mathrsfs}
\usepackage{xcolor}
\usepackage{tikz}
\usepackage{pgfplots}
\usepackage{esvect}
\usepackage{xspace}
\usepackage{array, multirow}
\usepackage{rotating, makecell}
\usepackage{enumitem}
\usepackage[calc]{adjustbox}

\newcommand*{\comments}{} 
\ifdefined\comments
\newcommand{\todo}[1]{\textcolor{red}{#1}}
\else
\newcommand{\todo}[1]{}
\fi

\newcommand{\irule}[2]%
   {\mkern-2mu\displaystyle\frac{#1}{\vphantom{,}#2}\mkern-2mu}
\newcommand{\irulelabel}[3]
{
\mkern-2mu
\begin{array}{ll}
\displaystyle\frac{#1}{\vphantom{,}#2} & #3
\end{array}
\mkern-2mu
}

\algnewcommand\algorithmicforeach{\textbf{for each}}
\algdef{S}[FOR]{ForEach}[1]{\algorithmicforeach\ #1\ \algorithmicdo}

\algnewcommand{\IfThenElse}[3]{
\State \algorithmicif\ #1\ \algorithmicthen\ #2\ \algorithmicelse\ #3}

\newcommand{\toolname}{\textsc{Relish}\xspace}
\newcommand{\eusolver}{\textsc{EUSolver}\xspace}

\newcommand{\denot}[1]{\llbracket #1 \rrbracket}
\newcommand{\set}[1]{\{ #1 \}}
\newcommand{\bigset}[1]{\big\{ #1 \big\}}
\newcommand{\functions}{\mathcal{F}}
\newcommand{\grammars}{\mathcal{G}}
\newcommand{\interps}{\mathcal{I}}
\newcommand{\lang}{\mathcal{L}}
\newcommand{\programs}{\mathcal{P}}
\newcommand{\occurs}{\mathcal{M}}
\newcommand{\bigO}{\mathcal{O}}

\newcommand{\fta}{\mathcal{A}}
\newcommand{\ftaStates}{Q}
\newcommand{\ftaAlphabet}{\Sigma}
\newcommand{\ftaFinal}{Q_f}
\newcommand{\ftaTrans}{\Delta}
\newcommand{\ftaRun}{\pi}
\newcommand{\ftastate}{q}
\newcommand{\ftastates}{\ftaStates}
\newcommand{\alphabet}{\ftaAlphabet}
\newcommand{\finalstates}{{\ftaFinal}}
\newcommand{\transitions}{\ftaTrans}

\newcommand{\hfta}{\mathcal{H}}
\newcommand{\hftaNodes}{V}
\newcommand{\hftaRoot}{v_r}
\newcommand{\hftaAnnot}{\Omega}
\newcommand{\hftaTrans}{\Lambda}

\newcommand{\htree}{\mathcal{T}}
\newcommand{\htreeNodes}{V}
\newcommand{\htreeEdges}{E}
\newcommand{\htreeRoot}{v_r}
\newcommand{\htreeAnnot}{\Upsilon}

\renewcommand{\dots}{\cdots}

\newcommand{\ex}{e}
\newcommand{\exs}{\vec{\ex}}

\newcommand{\inp}{{\emph{in}}}
\newcommand{\out}{{\emph{out}}}

\newcommand{\grammar}{G}
\newcommand{\terminals}{T}
\newcommand{\nonterminals}{N}
\newcommand{\productions}{P}
\newcommand{\outputsymbol}{s_0}
\newcommand{\inputsymbol}{x}

\newcommand{\semantics}[1]{\llbracket{#1}\rrbracket}

\begin{document}

\title{Relational Program Synthesis}         


\author{Yuepeng Wang}
\affiliation{
  \institution{University of Texas at Austin}
  \country{USA}
}
\email{ypwang@cs.utexas.edu}

\author{Xinyu Wang}
\affiliation{
  \institution{University of Texas at Austin}
  \country{USA}
}
\email{xwang@cs.utexas.edu}

\author{Isil Dillig}
\affiliation{
  \institution{University of Texas at Austin}
  \country{USA}
}
\email{isil@cs.utexas.edu}

\begin{abstract}

This paper proposes \emph{relational program synthesis}, a new problem that concerns synthesizing one or more programs that collectively satisfy a \emph{relational specification}. As a dual of relational program verification, relational program synthesis is an important problem that has many practical applications, such as automated program inversion and automatic generation of comparators. However, this relational synthesis problem introduces new challenges over its non-relational counterpart due to the combinatorially larger search space.  
As a first step towards solving this problem, this paper presents a synthesis technique that combines the counterexample-guided inductive synthesis framework with a novel inductive synthesis algorithm that is based on \emph{relational version space learning}.  
We have implemented the proposed technique in a framework called \toolname, which can be instantiated to different application domains by providing a suitable domain-specific language and the relevant relational specification. We have used the \toolname framework to build relational synthesizers to automatically generate string encoders/decoders as well as  comparators, and we evaluate our tool on several benchmarks taken from prior work and online forums. Our experimental results show that the proposed technique can solve almost all of these benchmarks and that it significantly outperforms \eusolver, a generic synthesis framework that won the general track of the most recent SyGuS competition. 

\end{abstract}



\maketitle

\section{Introduction}\label{sec:intro}

\emph{Relational properties} describe requirements on the interaction between multiple  programs or different runs of the same program. Examples of relational properties include the following:

\begin{itemize}
\item \emph{Equivalence:} Are two programs $P_1, P_2$ observationally equivalent? (i.e., $\forall \vec{x}.\  P_1(\vec{x}) = P_2(\vec{x})$)
\item \emph{Inversion:} Given two programs $P_1, P_2$, are they inverses of each other? (i.e., $\forall x. \ P_2(P_1(x)) = x)$
\item \emph{Non-interference:} Given a program $P$ with two types of inputs, namely \emph{low} (public) input $\vec{l}$ and \emph{high} (secret) input $\vec{h}$, does $P$ produce the same output when run on the same low input $\vec{l}$ but two different high inputs $\vec{h_1}, \vec{h_2}$? (i.e., $\forall \vec{l}, \vec{h_1}, \vec{h_2}. \ P(\vec{l}, \vec{h_1}) = P(\vec{l}, \vec{h_2}))$
\item \emph{Transitivity:} Given a program $P$ that returns a boolean value, does $P$ obey transitivity? (i.e., $\forall x,y,z. \ P(x,y) \land P(y,z) \Rightarrow P(x, z)$)
\end{itemize}
Observe that the first two properties listed above relate  different programs, while the latter two relate multiple runs of the same program. 

Due to their importance in a wide range of application domains, relational properties have received significant attention from the program verification community. For example, prior papers propose novel \emph{program logics} for verifying relational properties~\cite{rhl,prhl,chl,qchl,relsep} or transform the relational verification problem to standard safety  by constructing so-called \emph{product programs}~\cite{product1,product-asymmetric,product3}.

In this paper, we consider the dual \emph{synthesis} problem of relational verification. That is, given a relational specification $\Psi$ relating  runs of $n$ programs $P_1, \ldots, P_n$, our goal is to automatically {synthesize} $n$ programs that satisfy $\Psi$. This \emph{relational synthesis} problem has a broad range of practical applications. For example, we can use relational synthesis to automatically generate comparators
that provably satisfy certain correctness requirements  such as transitivity and anti-symmetry. As another example, we can use relational synthesis to solve the \emph{program inversion} problem where the goal is to generate a program $P'$ that is an inverse of another program $P$~\cite{inversion-loris,inversion-swarat}. Furthermore, since many automated repair techniques rely on program synthesis~\cite{repair1,repair2}, we believe that relational synthesis could also be useful for repairing programs that violate a relational property like non-interference.

Solving the relational synthesis problem introduces new challenges over its non-relational counterpart due to the combinatorially larger search space. In particular,  
a naive algorithm that simply enumerates combinations of programs and then checks their correctness with respect to the relational specification is unlikely to scale.
Instead, we need to design novel \emph{relational} synthesis algorithms 
to {efficiently} search for \emph{tuples of programs} that collectively satisfy the given specification.

We solve this challenge by introducing a novel synthesis algorithm that learns \emph{relational version spaces (RVS)}, a generalization of the notion of \emph{version space} utilized in prior work~\cite{vsa,flashfill,dace}. 
Similar to other synthesis algorithms based on counterexample-guided inductive synthesis (CEGIS)~\cite{cegis}, our method also alternates between inductive synthesis and verification; but the counterexamples returned by the verifier are relational in nature. Specifically, given a set of \emph{relational counterexamples}, such as $f(1) = g(1)$ or $f(g(1),g(2))=f(3)$, our inductive synthesizer  compactly represents \emph{tuples} of programs  and efficiently searches for their implementations that satisfy all relational counterexamples.

In more detail, our relational version space learning algorithm is based on the novel concept of \emph{hierarchical finite tree automata (HFTA)}. Specifically, an HFTA is a hierarchical collection of finite tree automata (FTAs), where each individual FTA represents  possible implementations of the different functions to be synthesized. Because  relational counterexamples can refer to compositions of functions (e.g., $f({g(1), g(2)})$), the HFTA representation allows us to compose different FTAs according to the hierarchical structure of  subterms in the relational counterexamples. Furthermore, our method constructs the HFTA  in such a way that tuples of programs that do not satisfy the examples are rejected and therefore excluded from the search space. Thus, the HFTA representation allows us to compactly represent those programs that are consistent with the relational examples.

We have implemented the proposed relational synthesis algorithm in a tool called \toolname~\footnote{\toolname stands for RELatIonal SyntHesis.} and evaluate it in the context of two different relational properties. 
First, we use \toolname to automatically synthesize encoder-decoder pairs that are provably inverses of each other. 
Second, we use \toolname to automatically generate comparators, which must satisfy three different relational properties, namely, anti-symmetry, transitivity, and totality. 
Our evaluation shows that \toolname can efficiently solve interesting benchmarks taken from previous literature and online forums. 
Our evaluation also shows that \toolname significantly outperforms \eusolver, a general-purpose synthesis tool that won the General Track of the most recent SyGuS competition for syntax-guided synthesis.

To summarize, this paper makes the following key contributions:

\begin{itemize}
\item We introduce the \emph{relational synthesis problem} and take a first step towards solving it.
\item We describe a \emph{relational version space} learning algorithm based on the concept of \emph{hierarchical finite tree automata} (HFTA).
\item We show how to construct HFTAs from relational examples expressed as ground formulas and describe an algorithm for {finding the desired accepting runs}.
\item We experimentally evaluate our approach in two different application domains and demonstrate the advantages of our relational version space learning approach over a state-of-the-art synthesizer based on enumerative search.
\end{itemize}

\section{Overview} \label{sec:overview}

\begin{figure}
\small
\begin{center}
\bgroup
\def\arraystretch{1.1}
\begin{tabular}{|c|c|}
\hline
 {\bf Property}  & {\bf Relational specification} \\ \hline
 Equivalence &   $\forall \vec{x}. \  f_1(\vec{x}) = f_2(\vec{x})$  \\ \hline
 Commutativity & $\forall x. \ f_1(f_2(x)) = f_2(f_1(x))$ \\ \hline
 Distributivity & $\forall x, y, z. \ f_2(f_1(x,y), z) = f_1(f_2(x,z), f_2(y,z))$ \\ \hline
 Associativity & $\forall x,y,z. \ f(f(x,y), z) = f(x, f(y,z))$ \\ \hline
 Anti-symmetry & $\forall x,y. \ (f(x,y)=\emph{true} \land x \neq y) \Rightarrow f(y, x)=\emph{false}$
\\ \hline
\end{tabular}
\egroup
\end{center}
\vspace{-8pt}
\caption{Examples of relational specifications for five relational properties.}\label{fig:ex-specs}
\vspace{-0.15in}
\end{figure}

In this section, we define the \emph{relational synthesis} problem and give a few motivating examples.

\subsection{Problem Statement}

The input to our synthesis algorithm is a \emph{relational specification} defined as follows:
\begin{definition}[\textbf{Relational specification}]\label{def:rel-spec}
A relational specification with functions $f_1, \ldots, f_n$ is a first-order sentence $\Psi = \forall \vec{x}. ~ \phi(\vec{x})$, where $\phi(\vec{x})$ is a quantifier-free formula with uninterpreted functions  $f_1, \ldots, f_n$.
\end{definition}

Fig.~\ref{fig:ex-specs} shows some familiar relational properties like distributivity and associativity  as well as  their corresponding specifications. Even though we refer to Definition~\ref{def:rel-spec} as a \emph{relational specification}, observe that it allows us to express combinations of both relational and non-relational properties. For instance, the specification $\forall x. \ (f(x) \geq 0  \land f(x) + g(x) = 0)$ imposes both a non-relational property on $f$ (namely, that all of its outputs must be non-negative) as well as the relational property that the outputs of $f$ and $g$ must always add up to zero.

\begin{definition}[\textbf{Relational program synthesis}]\label{def:rel-synth}
Given a set of $n$ function symbols $\functions = \set{f_1, \ldots, f_n}$, their corresponding domain-specific languages $\set{L_1, \ldots, L_n}$, and a relational specification $\Psi$, the relational program synthesis problem is to find an interpretation $\interps$ for $\functions$ such that:
\begin{enumerate} 
    \item For every function symbol $f_i \in \functions$, $\interps(f_i)$ is a program in $f_i$'s DSL $L_i$ (i.e., $\interps(f_i) \in L_i$).
    \item The interpretation $\interps$ satisfies the relational specification $\Psi$ (i.e., $\interps(f_1), \ldots, \interps(f_n) \models \Psi$).
\end{enumerate} 
\end{definition}

\subsection{Motivating Examples}\label{sec:motivating-ex}
We now illustrate the practical relevance of the relational program synthesis problem through several real-world programming scenarios.

\begin{example}[\textbf{String encoders and decoders}] \label{ex:codec}
Consider a programmer who needs to implement a Base64 encoder \texttt{encode(x)} and its corresponding decoder \texttt{decoder(x)} for any Unicode string \texttt{x}. For example, according to Wikipedia~\footnote{\url{https://en.wikipedia.org/wiki/Base64}}, the encoder should  transform the string ``\texttt{Man}'' into ``\texttt{TWFu}'', ``\texttt{Ma}'' into ``\texttt{TWE=}'', and ``\texttt{M}'' to ``\texttt{TQ==}''. In addition, applying the decoder to the encoded string should yield the original string. Implementing this encoder/decoder  pair is a relational synthesis problem in which the relational specification is the following:
\[
\small
\begin{array}{lll}
    & \texttt{encode(} \text{``} \texttt{Man} \text{''} \texttt{)} = \text{``} \texttt{TWFu} \text{''} \land \texttt{encode(} \text{``} \texttt{Ma} \text{''} \texttt{)} = \text{``} \texttt{TWE=} \text{''} \land \texttt{encode(} \text{``} \texttt{M} \text{''} \texttt{)} = \text{``} \texttt{TQ==} \text{''} & (\emph{input-output examples}) \\
\land & \forall x.~ \texttt{decode(encode($x$))} = x & (\emph{inversion}) \\
\end{array}
\]
Here, the first part  (i.e., the first line) of the specification gives three input-output examples for \texttt{encode}, and the second part states that \texttt{decode} must be the inverse of \texttt{encode}. \toolname can automatically synthesize the correct Base64 encoder and decoder from this specification using a DSL targeted for this domain (see Section~\ref{sec:encoder}).

\end{example}

\begin{example}[\textbf{Comparators}] \label{ex:comparator}
Consider a programmer who needs to implement a comparator for sorting an array of integers according to the number of occurrences of the number \texttt{5}~\footnote{\url{https://stackoverflow.com/questions/19231727/sort-array-based-on-number-of-character-occurrences}}.
Specifically, \texttt{compare(x,y)} should return \texttt{-1} (resp. \texttt{1}) if \texttt{x} (resp. \texttt{y}) contains less \texttt{5}'s than \texttt{y} (resp. \texttt{x}), and ties should be broken based on the actual values of the integers. For instance, sorting the array [``\texttt{24}'',``\texttt{15}'',``\texttt{55}'',``\texttt{101}'',``\texttt{555}''] using this comparator should yield array [``\texttt{24}'',``\texttt{101}'',``\texttt{15}'',``\texttt{55}'',``\texttt{555}'']. Furthermore, since the comparator must define a total order, its implementation should satisfy reflexivity, anti-symmetry, transitivity, and totality.  The problem of generating a suitable \texttt{compare} method is again a relational synthesis problem and can be defined using the following specification:
\[
\small
\begin{array}{lll}
    & \texttt{compare(} \text{``} \texttt{24} \text{''}, \text{``} \texttt{15} \text{''} \texttt{)} = \texttt{-1} \land \texttt{compare(} \text{``} \texttt{101} \text{''}, \text{``} \texttt{24} \text{''} \texttt{)} = \texttt{1} \land \ldots & (\emph{input-output examples}) \\
    \land & \forall x. ~ \texttt{compare($x,x$)} = \texttt{0} & (\emph{reflexivity}) \\
\land & \forall x, y. ~ \emph{sgn}(\texttt{compare($x,y$)}) = -\emph{sgn}(\texttt{compare($y, x$)}) & (\emph{anti-symmetry}) \\
\land & \forall x, y, z.~ \texttt{compare($x, y$)} > \texttt{0} \land \texttt{compare($y, z$)} > \texttt{0} \Rightarrow \texttt{compare($x, z$)} > \texttt{0} & (\emph{transitivity}) \\ 
\land & \forall x, y, z.~ \texttt{compare($x, y$)} = \texttt{0} \Rightarrow \emph{sgn}(\texttt{compare($x, z$)}) = \emph{sgn}(\texttt{compare($y, z$)}) & (\emph{totality}) \\
\end{array}
\]

A solution to this problem is given by the following implementation of \texttt{compare}:
\[
\small
\begin{BVerbatim}
let a = intCompare (countChar x '5') (countChar y '5') in
    if a != 0 then a else intCompare (toInt x) (toInt y)
\end{BVerbatim}
\]
where {\tt countChar} function returns the number of occurrences of a character in the input string, and {\tt intCompare} is the standard comparator on integers.

\end{example}

\begin{example}[\textbf{Equals and hashcode}]
A common programming task is to implement {\tt equals} and {\tt hashcode} methods for a given class.  These functions are closely related because {\tt equals(x,y)=true} implies that the hash codes of {\tt x} and {\tt y} must be the same. Relational synthesis can be used to simultaneously generate implementations of {\tt equals} and {\tt hashcode}. For example, consider an {\tt ExperimentResults} class that internally maintains an array of numbers  where negative integers indicate an anomaly (i.e., failed experiment) and should be ignored when comparing the results of two experiments. For instance, the results \texttt{[23.5,-1,34.7]} and \texttt{[23.5,34.7]} should be equal whereas \texttt{[23.5,34.7]} and \texttt{[34.7,23.5]} should not. The programmer can use a relational synthesizer to generate {\tt equals} and {\tt hashcode} implementations by providing the following specification:
\[
\small
\begin{array}{lll}
    & \texttt{equals([23.5,-1,34.7]}, \texttt{[23.5,34.7])} = \texttt{true} & (\emph{input-output examples}) \\
    \land  & \texttt{equals([23.5,34.7]}, \texttt{[34.7,23.5])} = \texttt{false} \land \ldots \\
\land & \forall x, y. \ \texttt{equals($x, y$)} \Rightarrow (\texttt{hashcode($x$)} = \texttt{hashcode($y$)}) & (\emph{equals-hashcode}) \\
\end{array}
\]

A possible solution consists of the following pair of implementations of {\tt equals} and {\tt hashcode}:
\[\small 
\begin{BVerbatim}
equals(x,y) : (filter (>= 0) x) == (filter (>= 0) y)
hashcode(x) : foldl (\u.\v. 31 * (u + v)) 0 (filter (>= 0) x)
\end{BVerbatim}
\]
\end{example}

\begin{example}[\textbf{Reducers}]  MapReduce is a software framework for writing applications that process large datasets in parallel. Users of this framework need to implement both a \emph{mapper} that maps input key/value pairs to intermediate ones as well as a \emph{reducer} which transforms intermediate values that share a key to a smaller set of values. An important requirement in the MapReduce framework is that the {\tt reduce} function must be associative. Now, consider the task of using the MapReduce framework to sum up sensor measurements obtained from an experiment.
In particular, each sensor value is either a real number or \texttt{None}, and the final result should be \texttt{None} if \emph{any} sensor value is \texttt{None}.  In order to implement this functionality, the user needs to write a reducer that takes sensor values \texttt{x}, \texttt{y} and returns their sum if neither of them is \texttt{None}, or returns \texttt{None} otherwise.  We can express this problem using the following relational specification:
\[
\small
\begin{array}{lll}
    & \texttt{reduce(1} , \texttt{2)} = \texttt{3} \land \texttt{reduce(1} , \texttt{None)} = \texttt{None} & (\emph{input-output examples}) \\
\land & \forall x, y, z. ~ \texttt{reduce(reduce($x, y$)$, z$)} = \texttt{reduce(}x, \texttt{reduce(} y, z \texttt{))} & (\emph{associativity}) \\
\end{array}
\]
The first part of the specification gives input-output examples to illustrate the desired functionality, and the latter part expresses the associativity requirement on {\tt reduce}. A possible solution to this relational synthesis problem is given by the following implementation:
\[\small
\begin{BVerbatim}
reduce(x,y) : if x == None then None
              else if y == None then None
              else x + y
\end{BVerbatim}
\]

\end{example}

\section{Preliminaries}

Since the rest of this paper requires knowledge of finite tree automata and their use in program synthesis, we first briefly review some background material.

\subsection{Finite Tree Automata}

A tree automaton is a type of state machine that recognizes trees rather than strings.  More formally, a (bottom-up) \emph{finite tree automaton} (FTA) is a tuple $\fta = (\ftaStates, \ftaAlphabet, \ftaFinal, \ftaTrans)$, where
\begin{itemize}
    \item $\ftaStates$ is a finite set of states.
    \item $\ftaAlphabet$ is an alphabet.
    \item $\ftaFinal \subseteq \ftaStates$ is a set of final states.
    \item $\ftaTrans \subseteq \ftaStates^* \times \ftaAlphabet \times \ftaStates$ is a set of transitions (or rewrite rules).
\end{itemize}

 Intuitively, a tree automaton $\fta$ recognizes a   term (i.e., tree) $t$ if we can rewrite $t$ into a final state $q \in \ftaFinal$ using the rewrite rules given by $\ftaTrans$.  

More formally, suppose we have a tree $t = (V, E, v_r)$ where $V$ is a set of nodes labeled with $\sigma \in \Sigma$, $E \subseteq V \times V$ is a set of edges, and $v_r \in V$ is the root node. A \emph{run} of an FTA $\fta = (\ftaStates, \ftaAlphabet, \ftaFinal, \ftaTrans)$ on $t$ is a mapping $\ftaRun: V \to \ftaStates$ compatible with $\ftaTrans$ (i.e., given node $v$ with label $\sigma$ and children $v_1, \ldots, v_n$ in the tree, the run $\ftaRun$ can only map $v, v_1, \ldots, v_n$ to $q, q_1, \ldots, q_n$  if there is a transition $\sigma(q_1, \ldots, q_n) \to q$ in $\ftaTrans$). 
Run $\ftaRun$ is said to be \emph{accepting} if it maps the root node of $t$ to a final state, and a tree $t$ is \emph{accepted by} an FTA $\fta$ if there exists an accepting run of $\fta$ on $t$. The \emph{language} $\lang(\fta)$ recognized by $\fta$ is the set of all those trees that are accepted by $\fta$.

\begin{wrapfigure}{R}{2.3in}
\centering
\includegraphics[scale=0.5]{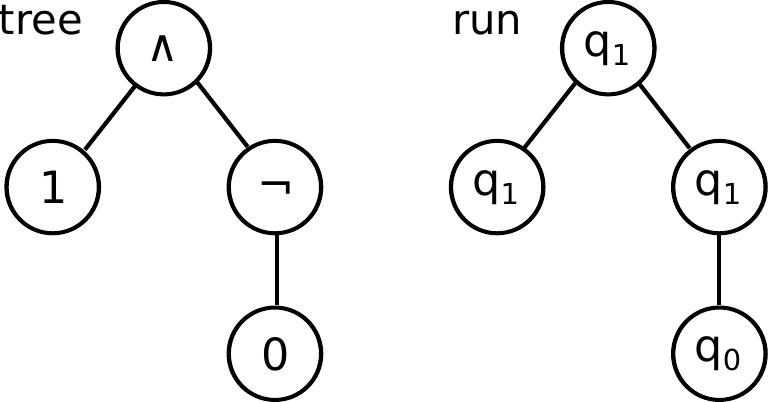}
\caption{Tree for $1 \land \neg 0$ and its accepting run.}
\label{fig:fta-ex}
\end{wrapfigure} 

\begin{example}
Consider a finite tree automaton $\fta = (\ftaStates, \ftaAlphabet, \ftaFinal, \ftaTrans)$ with states $\ftaStates = \set{q_0, q_1}$, alphabet $\ftaAlphabet = \set{0, 1, \neg, \land}$, final states $\ftaFinal = \set{q_1}$, and transitions $\ftaTrans$:
\[\small 
\begin{array}{l l l l}
1 \to q_1 & 0 \to q_0 & \land(q_0, q_0) \to q_0 & \land(q_0, q_1) \to q_0 \\
\neg(q_0) \to q_1 & \neg(q_1) \to q_0 & \land(q_1, q_0) \to q_0 & \land(q_1, q_1) \to q_1 \\
\end{array}
\]

Intuitively, the states of this FTA correspond to boolean constants (i.e., $q_0, q_1$ represent false and true respectively), and the transitions define the semantics of the boolean connectives $\neg$ and $\land$. Since the final state is $q_1$, the FTA accepts all boolean formulas (containing only $\land$ and $\neg$) that evaluate to $\emph{true}$. For example, $1 \land \neg 0$ is accepted by $\fta$, and the corresponding tree and its accepting run are shown in Fig.~\ref{fig:fta-ex}.
\end{example}

\subsection{Example-based Synthesis using FTAs}\label{sec:fta-construct}

Since our relational synthesis algorithm leverages prior work on example-based synthesis using FTAs~\cite{dace,syngar}, we briefly review how FTAs can be used for program synthesis.

At a high-level, the idea is to build an FTA that accepts \emph{exactly} the ASTs of those DSL programs that are consistent with the given input-output examples. The states of this FTA correspond to concrete values, and the transitions are constructed using the DSL's semantics. 
In particular, the FTA states correspond to output values of DSL programs on the input examples, and the final states of the FTA are determined by the given output examples. 
Once this FTA is constructed, the synthesis task boils down to finding an accepting run of the FTA. 
A key advantage of using FTAs for synthesis is to enable search space reduction by allowing sharing between programs that have the same input-output behavior.  

\begin{figure}[!t]
\small
\[
\begin{array}{cr}
\begin{array}{cc}
\irule{
\ex_{in}= (e_1, \ldots, e_n)
}{
\grammar, \ex \ \  \vdash \ \  \ftastate_{\inputsymbol_i}^{{\ex_i}} \in \ftastates, \ \  x_i \rightarrow q_{x_i}^{e_i} \in \transitions
} \ \ {\rm (Input)}
\ \ \ \ \ \ \ \ \ \ \ \ \ \ \ \ \ \ \ \ \ \ \  \ 
\irule{
}{
\grammar, \ex \ \ \vdash \ \   \ftastate_{\outputsymbol}^{\ex_\out}  \in \finalstates
} \ \  {\rm (Output)}
\end{array}
\\ \ \\
\begin{array}{l}
\irule{
\begin{array}{c}
(s \rightarrow \sigma(s_1, \ldots, s_n)) \in P \quad G, e \vdash q_{s_1}^{{c_1}} \in \ftastates, \ldots, q_{s_n}^{{c_n}} \in \ftastates 
\quad
c = \semantics{\sigma(c_{1}, \ldots, c_{n})}
\quad 
\end{array}
}{
\grammar = (T, \nonterminals, \productions, \outputsymbol), \ex \ \ \vdash  \ \ 
\ftastate_{s}^{{c}} \in \ftastates, \ \  \sigma(q_{s_1}^{{c_1}}, \ldots, q_{s_n}^{{c_n}}) \rightarrow \ftastate_{s}^{{c}}  \in \transitions
}\ \ {\rm (Prod)}
\end{array}
\end{array}
\]
\vspace{-10pt}
\caption{Rules for constructing FTA $\fta = (\ftastates, \alphabet, \finalstates, \transitions)$ given example $\ex = (\ex_\inp, \ex_\out)$ and grammar $\grammar =   (\terminals, \nonterminals, \productions, \outputsymbol)$. The alphabet $\alphabet$ of the FTA is exactly the terminals $T$ of $\grammar$.  }
\label{fig:fta-rules}
\vspace{-10pt}
\end{figure}

In more detail, suppose we are given a set of input-output examples $\exs$ and a context-free grammar $\grammar$ describing the target domain-specific language. 
We assume that $\grammar$ is of the form $(T, \nonterminals, \productions, \outputsymbol)$  where $T$ and $\nonterminals$ are the terminal and non-terminal symbols  respectively, $\productions$ is a set of productions of the form $s \rightarrow \sigma(s_1, \dots, s_n)$,
and $\outputsymbol \in \nonterminals$ is the topmost non-terminal (start symbol) in $\grammar$. We also assume that the program to be synthesized takes arguments  $x_1, \ldots, x_n$. 

Fig. \ref{fig:fta-rules} reviews the construction rules for a single example $\ex = (\ex_\inp, \ex_\out)$~\cite{syngar}. In particular, the Input rule creates $n$ initial states $q_{x_1}^{\ex_1}, \ldots, q_{x_n}^{\ex_n}$ for input example $\ex_\inp = (\ex_1, \ldots, \ex_n)$.  We then iteratively use the Prod rule to generate new states and transitions using the productions in grammar $\grammar$ and the concrete semantics of the DSL constructs associated with the production. Finally, according  to the Output rule, the only final state is $q_{s_0}^c$ where $s_0$ is the start symbol and $c$ is the output example $\ex_\out$. Observe that we can build an FTA for a \emph{set} of input-output examples by constructing the FTA for each example individually and then taking their intersection using standard techniques~\cite{tata}. 

\begin{wrapfigure}{R}{0.4\textwidth}
\vspace{-12pt}
\centering
\includegraphics[scale=0.7]{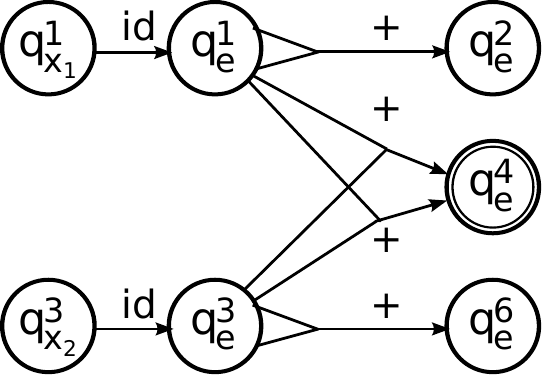}
\caption{FTA for input-output $(1, 3) \to 4$.}
\label{fig:fta-plus}
\vspace{-10pt}
\end{wrapfigure}

\vspace{2pt}
\noindent
{\bf \emph{Remark.}}  Since the rules in Fig.~\ref{fig:fta-rules} may generate an infinite set of states and transitions, the FTA is constructed by applying the Prod rule a finite number of times. The number of applications of the Prod rule is determined by a bound on the AST depth  of the synthesized program.

\begin{example}
Consider a very simple DSL specified by the following CFG, where  the notation $s \rightarrow s'$ is short-hand for $s \rightarrow \emph{id}(s')$ for a unary identity function $\emph{id}$:
\[ \hspace*{-160pt} e \ \ \rightarrow \ \ x_1 ~|~ x_2 ~|~ e + e \]
Suppose  we want to find a program (of size at most 3) that is consistent with the input-output example $(1, 3) \rightarrow 4$. Using the rules from Fig.~\ref{fig:fta-rules}, we can construct $\fta = (\ftastates, \alphabet, \finalstates, \transitions)$ with states $\ftastates = \set{q^1_{x_1}, q^3_{x_2}, q^1_e, q^2_e, q^3_e, q^4_e, q^6_e}$, alphabet $\alphabet = \set{x_1, x_2, \emph{id}, +}$, final states $\finalstates = \set{q^4_e}$, and transitions~$\transitions$:
\[\small 
\begin{array}{llll}
x_1 \to q^1_{x_1} & x_2 \to q^3_{x_2} & \emph{id}(q^1_{x_1}) \rightarrow q^1_e &  \emph{id}(q^3_{x_2}) \rightarrow q^3_e\\
+(q^1_e, q^1_e) \to q^2_e & +(q^1_e, q^3_e) \to q^4_e & +(q^3_e, q^1_e) \to q^4_e & +(q^3_e, q^3_e) \to q^6_e \\
\end{array}
\]
For example, the FTA contains the transition $+(q_e^3, q_e^1) \rightarrow q_e^4$ because the grammar contains a production $e \rightarrow e+e$ and we have $3+1 = 4$.  Since it is sometimes convenient to visualize FTAs as hypergraphs, Fig.~\ref{fig:fta-plus} shows a hypergraph representation of this FTA, using circles to indicate states, double circles to indicate final states, and labeled (hyper-)edges  to represent transitions.
The transitions of nullary alphabet symbols (i.e. $x_1 \to q^1_{x_1}, x_2 \to q^3_{x_2}$) are omitted in the hypergraph for brevity.
Note that the only two programs that are accepted by this FTA are $x_1 + x_2$ and $x_2 + x_1$.
\end{example}

\section{Hierarchical Finite Tree Automata}

As mentioned in Section~\ref{sec:intro}, our relational synthesis technique uses a version space learning algorithm that is based on the novel concept of \emph{hierarchical finite tree automata (HFTA)}. Thus, we first introduce HFTAs in this section before describing our relational synthesis algorithm.

A hierarchical finite tree automaton (HFTA) is a tree in which each node is annotated with a finite tree automaton (FTA). More formally, an HFTA is a tuple $\hfta = (\hftaNodes, \hftaAnnot, \hftaRoot, \hftaTrans)$ where
\begin{itemize}
\item 
$\hftaNodes$ is a finite set of nodes.
\item 
$\hftaAnnot$ is a mapping that maps each node $v \in \hftaNodes$ to a finite tree automaton.
\item 
$\hftaRoot \in \hftaNodes$ is the root node.
\item 
$\hftaTrans \subseteq \emph{States}(\emph{Range}(\hftaAnnot)) \times \emph{States}(\emph{Range}(\hftaAnnot))$ is a set of inter-FTA transitions.
\end{itemize}

Intuitively, an HFTA is a tree-structured (or hierarchical) collection of FTAs, where $\hftaTrans$ corresponds to the edges of the tree and specifies how to transition from a state of a child FTA to a state of its parent.
{
For instance, the left-hand side of Fig.~\ref{fig:hfta-ex} shows an HFTA, where the inter-FTA transitions  $\hftaTrans$ (indicated by dashed lines) correspond to the edges $(v_3, v_1)$ and $(v_3, v_2)$ in the tree.
}

Just as FTAs accept trees, HFTAs accept \emph{hierarchical trees}. 
Intuitively, as depicted in the right-hand side of Fig.~\ref{fig:hfta-ex}, a hierarchical tree organizes a collection of trees in a tree-structured (i.e., hierarchical) manner. 
More formally, a hierarchical tree is of the form $\htree = (\htreeNodes, \htreeAnnot, \htreeRoot, \htreeEdges)$, where:
\begin{itemize}
\item 
$\htreeNodes$ is a finite set of nodes.
\item 
$\htreeAnnot$ is a mapping that maps each node $v \in \htreeNodes$ to a tree.
\item 
$\htreeRoot \in \htreeNodes$ is the root node.
\item 
{$\htreeEdges \subseteq \htreeNodes \times \htreeNodes$ is a set of edges. } 
\end{itemize}

\begin{figure}[!t]
\centering
\includegraphics[scale=0.42]{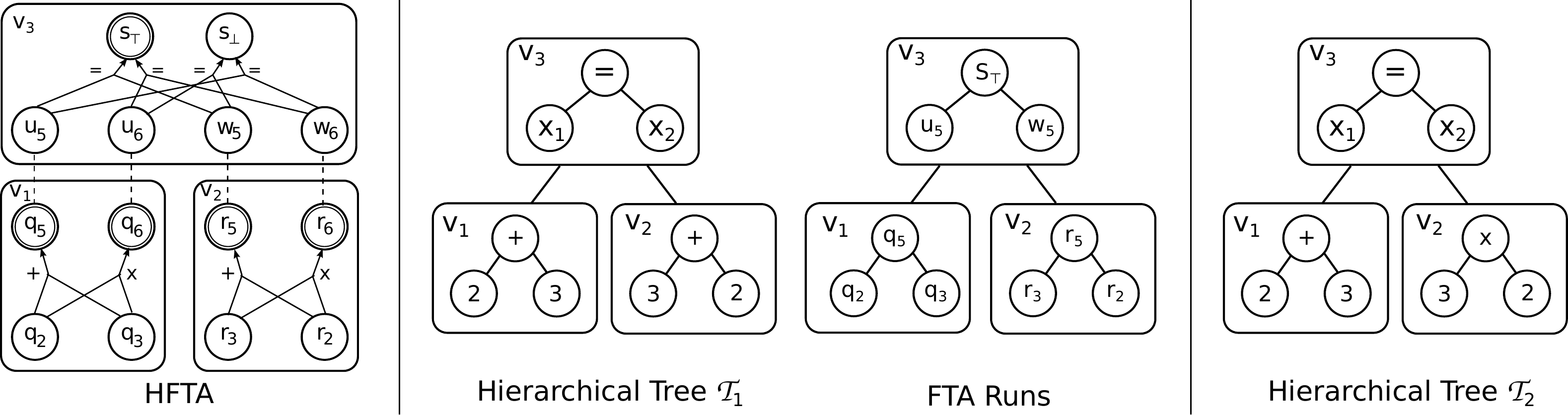}
\vspace{-7pt}
\caption{Example HFTA and hierarchical trees.
Left-hand side shows HFTA $\hfta = (\hftaNodes, \hftaAnnot, v_3, \hftaTrans)$, where nodes $\hftaNodes = \set{v_1, v_2, v_3}$ are represented by rectangles, and annotated FTAs are represented as hypergraphs inside the corresponding rectangles. Specifically, circles in the hypergraph correspond to FTA states, double circles indicate final states, and labeled (hyper-)edges correspond to transitions. Inter-FTA transitions $\hftaTrans$ are represented as dashed lines.
An example hierarchical tree $\htree_1 = (\htreeNodes, \htreeAnnot, v_3, \htreeEdges)$ that is \emph{accepted} by $\hfta$ is shown in the middle, where nodes $\htreeNodes = \set{v_1, v_2, v_3}$ are represented by rectangles, the tree annotation $\htreeAnnot(v_i)$ is depicted inside the corresponding rectangle $v_i$, and edges are $\htreeEdges = \set{(v_3, v_1), (v_3, v_2)}$. ``FTA runs'' shows the runs of the FTAs for each $v_i$ in $\hfta$ on the corresponding tree in $\htree_1$. 
On the right-hand side, we show another example of a hierarchical tree $\htree_2$ that is \emph{not accepted} by $\hfta$. 
}
\label{fig:hfta-ex}
\vspace{-5pt}
\end{figure}

We now define what it means for an HFTA to \emph{accept} a hierarchical tree:

\begin{definition}\label{def:hfta-accept}
A hierarchical tree $\htree = (\htreeNodes, \htreeAnnot, \htreeRoot, \htreeEdges)$ is accepted by an HFTA $\hfta = (\hftaNodes, \hftaAnnot, \hftaRoot, \hftaTrans)$~iff:
\begin{itemize}
\item 
For any node $v \in \htreeNodes$, the tree $\htreeAnnot(v)$ is accepted by FTA $\hftaAnnot(v)$. 
\item 
For any edge $(v, v') \in E$, there is an accepting run $\pi$ of FTA $\hftaAnnot(v)$ on tree $\htreeAnnot(v)$ and an accepting run $\pi'$ of $\hftaAnnot(v')$ on $\htreeAnnot(v')$ such that there exists a unique leaf node $l \in \htreeAnnot(v)$ where we have $\pi' \big(\emph{Root}(\htreeAnnot(v')) \big) \rightarrow \pi(l) \in \hftaTrans$. 
\end{itemize}
\end{definition}

In other words, a hierarchical tree $\htree$ is accepted by an HFTA $\hfta$ if the tree at each node $v$ of $\htree$ is accepted by the corresponding FTA $\hftaAnnot(v)$ of $\hfta$, and the runs of the individual FTAs can be ``stitched'' together according to the inter-FTA transitions of  $\hfta$.
The language of an HFTA $\hfta$, denoted $\lang(\hfta)$, is the set of all hierarchical trees accepted by $\hfta$.

\begin{example}

Consider the HFTA $\hfta = (\hftaNodes, \hftaAnnot, v_3, \hftaTrans)$ shown on the left-hand side of Fig.~\ref{fig:hfta-ex}. This HFTA has nodes $\hftaNodes = \set{v_1, v_2, v_3}$ annotated with FTAs as follows:
\[ \hftaAnnot = \set{v_1 \mapsto \fta_1, v_2 \mapsto \fta_3, v_3 \mapsto \fta_3}\]
where:

$\fta_1 = (\set{q_2, q_3, q_5, q_6}, \set{2, 3, +, \times}, \set{q_5, q_6}, \ftaTrans_1)$ with transitions $\ftaTrans_1$:
\[\small 
\begin{array}{llll}
2 \to q_2 & 3 \to q_3 & +(q_2, q_3) \to q_5 & \times(q_2, q_3) \to q_6 \\
\end{array}
\] 

$\fta_2 = (\set{r_2, r_3, r_5, r_6}, \set{2, 3, +, \times}, \set{r_5, r_6}, \ftaTrans_2)$ with transitions $\ftaTrans_2$:
\[\small
\begin{array}{lllll}
2 \to r_2 & 3 \to r_3 & +(r_3, r_2) \to r_5 & \times(r_3, r_2) \to r_6 \\
\end{array}
\]

$\fta_3 = (\set{u_5, u_6, w_5, w_6, s_\top, s_\bot}, \set{x_1, x_2, =}, \set{s_\top}, \ftaTrans_3)$ with transitions $\ftaTrans_3$:
\[\small 
\begin{array}{llll}
x_1 \to u_5 & x_1 \to u_6 & =(u_5, w_5) \to s_\top & =(u_6, w_5) \to s_\bot \\
x_2 \to w_5 & x_2 \to w_6 & =(u_5, w_6) \to s_\bot & =(u_6, w_6) \to s_\top \\
\end{array}
\]
Finally, $\hftaTrans$ includes the following four transitions: 
\[\small 
\begin{array}{llll}
q_5 \to u_5 \ \ \ \ \ \ \ & q_6 \to u_6 \ \ \ \ \ \ \ & r_5 \to w_5 \ \ \ \ \ \ \ & r_6 \to w_6 
\end{array}
\]

Intuitively, this HFTA accepts ground first-order formulas of the form $t_1 = t_2$ where $t_1$ (resp. $t_2$) is a ground term accepted by $\fta_1$ (resp. $\fta_2$) such that $t_1$ and $t_2$ are equal.
For example, $t_1 = 2 + 3$ and $t_2 = 3 + 2$ are collectively accepted. However, $t_1 = 2 + 3$ and $t_2 = 3 \times 2$ are not accepted because $t_1$ is not equal to $t_2$.
Fig.~\ref{fig:hfta-ex} shows a hierarchical tree $\htree_1$ that is accepted by $\hfta$ and another one $\htree_2$ that is rejected by $\hfta$. 
\end{example}

\section{Relational Program Synthesis using HFTA} \label{sec:synth}

In this section, we  describe  the high-level structure of our relational synthesis algorithm and explain its sub-procedures in the following subsections.

Our top-level synthesis algorithm is shown in Algorithm~\ref{algo:cegis}. The procedure {\sc Synthesize} takes as input a tuple of function symbols $\functions = (f_1, \ldots, f_n)$, their corresponding context-free grammars $\grammars = (G_1, \ldots, G_n)$, and a relational specification $\Psi = \forall \vec{x}. \ \phi(\vec{x})$. 
The output is a mapping $\programs$ from each function symbol $f_i$ to a program in grammar $G_i$ such that these programs collectively satisfy the given relational specification $\Psi$, or \emph{null} if no such solution exists. 

As mentioned in Section~\ref{sec:intro}, our synthesis algorithm is based on the framework of counterexample-guided inductive synthesis (CEGIS)~\cite{cegis}. 
Specifically, given candidate programs $\programs$, the verifier checks whether these programs satisfy the relational specification $\Psi$. If so, the CEGIS loop terminates with $\programs$ as a solution (line 6). 
Otherwise, we obtain a \emph{relational counterexample} in the form of a \emph{ground formula} and use it to strengthen the current \emph{ground relational specification} $\Phi$ (line 7). 
In particular, a relational counterexample can be obtained by instantiating the quantified variables in the relational specification $\Psi$ with concrete inputs that violate $\Psi$.
Then, in the \emph{inductive synthesis phase}, the algorithm finds candidates $\programs$ that satisfy the new ground relational specification $\Phi$ (lines 8-11) and repeats the aforementioned process. 
Initially, the candidate programs $\programs$ are chosen at random, and the ground relational specification $\Phi$ is \emph{true}.

Since the verification procedure is not a contribution of this paper, the rest of this section focuses on the inductive synthesis phase  which proceeds in three steps:
\begin{enumerate}[leftmargin=*]
\item {\bf \emph{Relaxation:}} 
Rather than directly taking $\Phi$ as the inductive specification, 
our algorithm first constructs a relaxation  $\Phi'$ of $\Phi$ (line 8) by replacing each occurrence of function $f_i$ in $\Phi$ with a fresh symbol ({e.g., $f_1(f_1(1)) = 1$ is converted to $f_{1,1}(f_{1,2}(1)) = 1$}). This strategy relaxes the requirement that different occurrences of a function symbol must have the same implementation and allows us to construct our relational version space in a \emph{compositional} way. 
\vspace{1pt}
\item {\bf \emph{HFTA construction:}} 
Given the relaxed specification $\Phi'$, the next step is to construct an HFTA whose language consists of exactly those programs that satisfy $\Phi'$ (line 9). 
Our HFTA construction follows the recursive decomposition of $\Phi'$ and allows us to compose the individual version spaces of each subterm in a way that the final HFTA is consistent with $\Phi'$. 
\vspace{1pt}
\item  {\bf \emph{Enforcing functional consistency:}} 
Since $\Phi'$ relaxes the original ground relational specification $\Phi$, the programs for different occurrences of the same function symbol in $\Phi$ that are accepted by the HFTA from Step (2) may be different. 
Thus, in order to guarantee functional consistency, our algorithm looks for programs that are accepted by the HFTA where all occurrences of the same function symbol correspond to the same program (lines 10-11).
\end{enumerate}

The following subsections discuss each of these three steps in more detail.

\begin{figure}[!t]
\vspace{-10pt}
\begin{algorithm}[H]
\caption{Top-level algorithm for inductive relational program synthesis.}
\label{algo:cegis}
\begin{algorithmic}[1]
\Procedure{\textsc{Synthesize}}{$\functions, \grammars, \Psi$}
\vspace{2pt}
\small 
\Statex 
\textbf{input:} Function symbols $\functions = (f_1, \ldots, f_n)$, grammars $\grammars = (G_1, \ldots, G_n)$, and specification $\Psi = \forall \vec{x}.~\phi(\vec{x})$. 
\Statex 
\textbf{output:} Mapping $\programs$ that maps function symbols in $\functions$ to programs. 
\vspace{2pt}
\State 
$\programs \gets \programs_{\text{random}}$; 
\State $\Phi \gets \emph{true}$; 
\While{\emph{true}}
\If{$\programs = \emph{null}$ } \Return $\emph{null}$; \EndIf
\If{$\textsf{Verify}(\programs, \Psi)$} \Return $\programs$; \EndIf 
\State $\Phi \gets \Phi \land \textsf{GetRelationalCounterexample}(\programs, \Psi)$; 
\vspace{1pt}
\State $(\Phi', \occurs) \gets \textsc{Relax}(\Phi, \functions)$;
\State $\hfta \gets \textsc{BuildHFTA}(\Phi', \grammars, \occurs)$;
\State $\programs_{\bot} \gets \set{f_1 \mapsto \bot, \ldots, f_n \mapsto \bot}$;
\State $\programs \gets \textsc{FindProgs}(\hfta, \programs_{\bot}, \occurs)$;
\EndWhile
\EndProcedure
\end{algorithmic}
\end{algorithm}
\vspace{-15pt}
\end{figure}

\subsection{Specification Relaxation} \label{sec:relax}

Given a ground relational specification $\Phi$ and a set of functions $\functions$ to be synthesized, the {\sc Relax} procedure generates a relaxed specification $\Phi'$ as well as a mapping $\occurs$ from each function $f \in \functions$ to a \emph{set} of fresh function symbols that represent different occurrences of $f$ in $\Phi$. As mentioned earlier, this relaxation strategy allows us to build a relational version space in a compositional way by composing the individual version spaces for each subterm. In particular, constructing an FTA for a function symbol $f \in \functions$ requires knowledge about the possible values of its arguments, which can introduce circular dependencies if we have multiple occurrences of $f$ (e.g., $f(f(1))$). The relaxed specification, however, eliminates such circular dependencies and allows us to build a relational version space in a compositional way.

We present our specification relaxation procedure, {\sc Relax}, in the form of inference rules shown in Fig.~\ref{fig:relax}. Specifically, Fig.~\ref{fig:relax} uses judgments of the form  $\functions \vdash \Phi \hookrightarrow (\Phi', \occurs)$, meaning that ground specification $\Phi$ is transformed into $\Phi'$ and $\occurs$ gives the mapping between the function symbols used in $\Phi$ and $\Phi'$. In particular, a mapping $f \mapsto \{ f_1, \ldots, f_n\}$ in $\occurs$ indicates that function symbols $f_1, \ldots, f_n$ in $\Phi'$ all correspond to the same function $f$ in $\Phi$.
It is worthwhile to point out that there are no inference rules for variables in Fig.~\ref{fig:relax} because specification $\Phi$ is a ground formula that does not contain any variables.

Since the relaxation procedure is pretty straightforward, we do not explain the rules in Fig.~\ref{fig:relax} in detail. At a high level, we recursively transform all subterms  and combine the results using a special union operator $\Cup$ defined as follows:
\[\small 
(\occurs_1 \Cup \occurs_2)(f) = 
\left \{
\begin{array}{ll}
\emptyset & {\rm if} f \not \in \emph{dom}(\occurs_1) \text{ and } f \not\in \emph{dom}(\occurs_2) \\
\occurs_i(f) & {\rm if} f  \in \emph{dom}(\occurs_i) \text{ and } f  \not \in \emph{dom}(\occurs_j) \\
\occurs_1(f) \cup \occurs_2(f) & {\rm otherwise} \\
\end{array}
\right .
\]

In the remainder of this paper, we also use the notation $\occurs^{-1}$ to denote the inverse of $\occurs$. That is, $\occurs^{-1}$ maps each function symbol in  the relaxed specification $\Phi'$ to a unique function symbol in the original specification $\Phi$ (i.e., $
f = \occurs^{-1}(f_i) \ \Leftrightarrow \ f_i \in \occurs(f)
$).

\begin{figure}[!t]
\centering
\small
\[
\begin{array}{c}
\irulelabel
{ }
{ \functions \vdash c \hookrightarrow (c, \emptyset) }
{ \textrm{(Const)} }

\ \ \ \ 

\irulelabel
{ \begin{array}{c}
f \in \functions \qquad f' ~ \textrm{is fresh} \qquad \functions \vdash t_i \hookrightarrow (t_i', \occurs_i) \qquad i = 1, \ldots, m \\
\end{array}}
{ \functions \vdash f(t_1, \ldots, t_m) \hookrightarrow (f'(t_1', \ldots, t_m'), ~ \Cup^{m}_{i=1} \occurs_i \Cup \set{f \mapsto \set{f'}}) }
{ \textrm{(Func)} }

\vspace{0.20in}
\\

\irulelabel
{ \begin{array}{c}
\functions \vdash \psi_1 \hookrightarrow (\psi_1', \occurs_1) \qquad \functions \vdash \psi_2 \hookrightarrow (\psi_2', \occurs_2) \\
\end{array}}
{ \functions \vdash \psi_1 ~ \emph{op} ~ \psi_2  \hookrightarrow (\psi_1' ~ \emph{op} ~ \psi_2', ~ \occurs_1 \Cup \occurs_2) }
{ \textrm{(Logical)} }

\qquad

\irulelabel
{ \begin{array}{c}
\functions \vdash \Phi \hookrightarrow (\Phi', \occurs) \\
\end{array}}
{ \functions \vdash \neg \Phi \hookrightarrow (\neg \Phi', \occurs) }
{ \textrm{(Neg)} }

\end{array}
\]
\caption{Inference rules for relaxing inductive relational specification $\Phi$ to $\Phi'$. $\occurs$ maps each function symbol in $\Phi$ to its occurrences in $\Phi'$. $\Cup$ is a special union operator for $\occurs$ and $\psi$ denotes either a term or a formula.}
\label{fig:relax}
\end{figure}

\subsection{HFTA Construction} \label{sec:build}

We now turn our attention to the relational version space learning algorithm using hierarchical finite tree automata. Given a relaxed relational specification $\Phi$, our algorithm builds an HFTA $\hfta$ that recognizes exactly those programs that satisfy $\Phi$. 
Specifically, each node in $\hfta$ corresponds to a subterm (or subformula) in $\Phi$. For instance, consider the subterm $f(g(1))$ where $f$ and $g$ are functions to be synthesized. In this case, the HFTA contains  a node  $v_1$ for $f(g(1))$ whose child $v_2$ represents the subterm $g(1)$. The inter-FTA transitions represent data flow from the children subterms to the parent term, and the FTA for each node $v$ is constructed according to the grammar of the top-level constructor of the subterm represented by $v$.

Fig.~\ref{fig:build-hfta} describes HFTA construction in more detail using the judgment $\grammars, \occurs \vdash \Phi \twoheadrightarrow \hfta$ where $\grammars$ describes the DSL for each function symbol~\footnote{If there is an interpreted function in the relational specification, we assume that a trivial grammar with a single production for that function is added to $\grammars$ by default.}
and $\occurs$ is the mapping constructed in Section~\ref{sec:relax}. The meaning of this judgment is that, under CFGs $\grammars$ and mapping $\occurs$, HFTA $\hfta$ represents exactly those programs that are consistent with $\Phi$.
In addition to the judgment $\grammars, \occurs \vdash \Phi \twoheadrightarrow \hfta$, Fig.~\ref{fig:build-hfta} also uses an auxiliary judgment of the form $\grammars, \occurs \vdash \Phi \leadsto \hfta$ to construct the HFTAs for the subformulas and subterms in the specification $\Phi$.

Let us now consider the HFTA construction rules from Fig.~\ref{fig:build-hfta} in more detail. The first rule, called Const, is the base case of our algorithm and builds an HFTA node for a constant $c$. Specifically, we introduce a fresh node $v_c$ and construct a trivial FTA $\fta$ that accepts only constant $c$.

The next rule, Func, shows how to build an HFTA for a function term $f(t_1, \ldots, t_m)$.
In this rule, we first recursively build the HFTAs $\hfta_i = (\hftaNodes{_i}, \hftaAnnot_i, v_{r_i}, \hftaTrans_i)$ for each subterm $t_i$. 
Next, we use a procedure called {\sc BuildFTA} to build an FTA $\fta$ for the function symbol $f$ that takes $m$ arguments $x_1, \dots, x_m$.
The {\sc BuildFTA} procedure is shown in Fig.~\ref{fig:buildFTA} and is a slight adaptation of the example-guided FTA construction technique described in Section~\ref{sec:fta-construct}: {Instead of assigning the input example to each argument $x_i$, our \textsc{BuildFTA} procedure assigns possible values that $x_i$ may take  based on the HFTA $\hfta_i$ for subterm $t_i$.
Specifically, let $Q_{f_i}$ denote the final states of the FTA $\fta_{f_i}$ associated with the root node of $\hfta_i$. If there is a final state $q_{s}^c $ in $Q_{f_i}$, this indicates that argument $x_i$ of $f$ may take value $c$; thus our {\sc BuildFTA} procedure adds a state $q_{x_i}^c$ in the FTA $\fta$ for $f$ (see the Input rule in Fig.~\ref{fig:buildFTA}).
The transitions and new states are built exactly as described in Section~\ref{sec:fta-construct} (see the Prod rule of Fig.~\ref{fig:buildFTA}), and we mark every state $q_{s_0}^c$ as a final state if $s_0$ is the start symbol in $f$'s grammar.
To avoid getting an FTA of infinite size, we impose a predefined bound on how many times the Prod rule can be applied when building the FTA.
Using this FTA $\fta$ for $f$ and the HFTAs $\hfta_i$ for $f$'s subterms, we now construct a bigger HFTA $\hfta$ as follows: First, we introduce a new node $v_f$ for function $f$ and annotate it with $\fta$. Second, we add appropriate inter-FTA transitions in $\hftaTrans$ to account for the binding of formal to actual parameters. Specifically, if $q_s^c$ is a final state of $\fta_{f_i}$, we add a transition $q_s^c \rightarrow q_{x_i}^c$ to $\hftaTrans$. Observe that our HFTA is compositional in that there is a separate node for every subterm, and the inter-FTA transitions account for the flow of information from each subterm to its parent.

\begin{figure}[!t]
\vspace{-5pt}
\centering
\small
\[
\begin{array}{c}

\irulelabel
{ \begin{array}{c}
\hftaNodes = \set{v_c} \qquad \fta = ( \big\{ q^{\denot{c}}_c \big\}, \set{c}, \big\{ q^{\denot{c}}_c \big\}, \big\{ c \to q^{\denot{c}}_c \big\} ) \qquad \hftaAnnot = \set{v_c \mapsto \fta} \\
\end{array}}
{ \grammars, \occurs \vdash c \leadsto (\hftaNodes, \hftaAnnot, v_c, \emptyset) }
{ \textrm{(Const)} }

\vspace{0.20in}
\\

\irulelabel
{ \begin{array}{c}
\grammars, \occurs \vdash t_i \leadsto (\hftaNodes_i, \hftaAnnot_i, v_{r_i}, \hftaTrans_i) \qquad \hftaAnnot_i(v_{r_i}) = (\ftaStates_i, \ftaAlphabet_i, Q_{f_i}, \ftaTrans_i) \qquad i = 1, \ldots, m \\
\hftaAnnot = \bigcup^{m}_{i=1} \hftaAnnot_i \cup \big\{ v_f \mapsto \textsc{BuildFTA} \big( \grammars(\occurs^{-1}(f)), [Q_{f_1}, \ldots, Q_{f_m}] \big) \big\} \\
\hftaTrans = \bigcup^{m}_{i=1} \hftaTrans_i \cup \big\{ q^c_s \to q^c_{x_i} \ \big| \ q^c_s \in Q_{f_i}, i \in [1, m] \big\} \\
\end{array}}
{ \grammars, \occurs \vdash f(t_1, \ldots, t_m) \leadsto (\bigcup^{m}_{i=1} \hftaNodes_i \cup \set{v_f}, \hftaAnnot, v_f, \hftaTrans) }
{ \textrm{(Func)} }

\vspace{0.20in}
\\

\hspace{-0.15in}
\irulelabel
{ \begin{array}{rl}
\grammars, \occurs \vdash \psi_i \leadsto (\hftaNodes_i, \hftaAnnot_i, v_{r_i}, \hftaTrans_i) & \hftaAnnot_i(v_{r_i}) = (\ftaStates_i, \ftaAlphabet_i, Q_{f_i}, \ftaTrans_i) \qquad i = 1, 2 \\
\ftaStates = \big\{ q^{\bot}_{s_0}, q^{\top}_{s_0} \big\} \cup \big\{ q^{c_i}_{x_i} \ \big| \ q^{c_i}_{s_i} \in Q_{f_i}, i = 1, 2 \big\} & \ftaTrans = \big\{ \emph{op}(q^{c_1}_{x_1}, q^{c_2}_{x_2}) \to q^c_{s_0} \ \big| \ c = \denot{\emph{op}}(c_1, c_2) \big\} \\
\hftaAnnot = \hftaAnnot_1 \!\cup\! \hftaAnnot_2 \!\cup\! \big\{ v_{\emph{op}} \mapsto (\ftaStates, \set{x_1, x_2, \emph{op}}, \big\{ q^{\bot}_{s_0}, q^{\top}_{s_0} \big\}, \ftaTrans) \big\} & \hspace{-0.08in} \hftaTrans = \hftaTrans_1 \!\cup\! \hftaTrans_2 \!\cup\! \big\{ q^{c}_{s_i} \to q^{c}_{x_i} \ \big| \ q^c_{s_i} \in Q_{f_i}, i = 1, 2 \big\} \\
\end{array}}
{ \grammars, \occurs \vdash \psi_1 ~ \emph{op} ~ \psi_2  \leadsto (\hftaNodes_1 \cup \hftaNodes_2 \cup \set{v_{\emph{op}}}, \hftaAnnot, v_{\emph{op}}, \hftaTrans) }
{ \hspace{-0.1in} \textrm{(Logical)} }

\vspace{0.20in}
\\

\irulelabel
{ \begin{array}{rl}
\grammars, \occurs \vdash \Phi \leadsto (\hftaNodes, \hftaAnnot, v_{r}, \hftaTrans) & \hftaAnnot(v_{r}) = (\ftaStates, \ftaAlphabet, Q_{f}, \ftaTrans) \\
\ftaStates^\prime = \big \{ q^{\bot}_{x_1}, q^{\top}_{x_1}, q^{\bot}_{s_0}, q^{\top}_{s_0} \big\} & \ftaTrans^\prime = \big\{ \neg(q^{\bot}_{x_1}) \to q^{\top}_{s_0}, \neg(q^{\top}_{x_1}) \to q^{\bot}_{s_0} \big\} \\
\hftaAnnot^\prime = \hftaAnnot \cup \big\{ v_\neg \mapsto \big( \ftaStates^\prime, \set{x_1, \neg}, \big\{ q^{\top}_{s_0}, q^{\bot}_{s_0} \big\}, \ftaTrans^\prime \big) \big\} & \hftaTrans^\prime = \hftaTrans \cup \big\{ q^{c}_{s} \to q^{c}_{x_1} \ \big| \ q^c_{s} \in Q_{f} \big\} \\
\end{array}}
{ \grammars, \occurs \vdash \neg \Phi \leadsto (\hftaNodes \cup \set{v_\neg}, \hftaAnnot^\prime, v_\neg, \hftaTrans^\prime) }
{ \textrm{(Neg)} }

\vspace{0.20in}
\\

\irulelabel
{ \begin{array}{c}
\grammars, \occurs \vdash \Phi \leadsto (\hftaNodes, \hftaAnnot, \hftaRoot, \hftaTrans) \qquad \hftaAnnot(\hftaRoot) = (\ftaStates, \ftaAlphabet, \ftaFinal, \ftaTrans) \\
\hftaAnnot' = \hftaAnnot \big[ \hftaRoot \mapsto (\ftaStates, \ftaAlphabet, \big\{ q^{\top}_{s_0} \big\}, \ftaTrans) \big] \\
\end{array}}
{ \grammars, \occurs \vdash \Phi \twoheadrightarrow (\hftaNodes, \hftaAnnot', \hftaRoot, \hftaTrans) }
{ \textrm{(Final)} }

\end{array}
\]
\caption{Rules for constructing HFTA $\hfta = (\hftaNodes, \hftaAnnot, \hftaRoot, \hftaTrans)$ for ground relational specification $\Phi$ from grammars $\grammars$, and function symbol mapping $\occurs$. Every node $v$ is a fresh node. $\psi$ denotes either a term or a formula.
Every variable $x_i$ in function $f$ is assumed to be annotated by $f$ to enforce its global uniqueness.
}
\label{fig:build-hfta}
\end{figure}

\begin{figure}[!t]
\vspace{-10pt}
\small
\[
\begin{array}{cr}
\begin{array}{ccc}
\irule{
q_s^c \in Q_i 
}{
\grammar, \vec{Q} \ \  \vdash \ \  \ftastate_{\inputsymbol_i}^{{c}} \in \ftastates, \ \ x_i \rightarrow q_{x_i}^c \in \transitions 
} \ \ {\rm (Input)} \quad \quad
\irule{
\grammar, \vec{Q} \vdash q_{s_0}^c \in Q
}{
\grammar, \vec{Q} \ \ \vdash \ \  q_{s_0}^c  \in \finalstates
} \ \  {\rm (Output)}
\end{array}
\\ \\
\irule{
\begin{array}{c}
\big( s \rightarrow \sigma(s_1, \ldots, s_n) \big) \in P \quad G, \vec{Q} \vdash q_{s_1}^{{c_1}} \in \ftastates, \ldots, q_{s_n}^{{c_n}} \in \ftastates 
\quad
c = \semantics{\sigma(c_{1}, \ldots, c_{n})}
\quad 
\end{array}
}{
\grammar = (T, \nonterminals, \productions, \outputsymbol), \vec{Q} \ \ \vdash  \ \ 
\ftastate_{s}^{{c}} \in \ftastates, ~ \sigma(q_{s_1}^{{c_1}}, \ldots, q_{s_n}^{{c_n}}) \rightarrow \ftastate_{s}^{{c}}  \in \transitions
}\ \ {\rm (Prod)}
\end{array}
\]
\caption{
Auxiliary \textsc{BuildFTA} procedure to construct an FTA $\fta = (\ftastates, \alphabet, \finalstates, \transitions)$ for function $f$ that takes $m$ arguments $x_1, \ldots, x_m$. $\grammar$ is a context-free grammar of function $f$. $\vec{Q} = (Q_1, \ldots, Q_m)$ is a vector where $Q_i$ represents the initial state set for argument $x_i$.
}
\label{fig:buildFTA}
\end{figure}

The next rule, called Logical, shows how to build an HFTA for a subterm $\psi_1 \ op \ \psi_2$ where \emph{op} is either a relation constant (e.g., $=, \leq$) or a binary logical connective (e.g., $\land, \leftrightarrow$). 
As in the previous rule, we first recursively construct the HFTAs $\hfta_i = (\hftaNodes{_i},\hftaAnnot_i, v_{r_i}, \hftaTrans_i)$ for the subterms $\psi_1$ and $\psi_2$. 
Next, we introduce a fresh HFTA node $v_{\emph{op}}$ for the subterm ``$\psi_1 \ op \ \psi_2$'' and construct its corresponding FTA $\fta$ using \emph{op}'s semantics. 
The inter-FTA transitions between $\fta$ and the FTAs annotating the root nodes of $\hfta_1$ and $\hfta_2$ are constructed in the same way as in the previous rule.~\footnote{In fact, the Logical and Neg rules could be viewed as special cases of the Func rule where \emph{op} and $\neg$ have a trivial grammar with a single production. However, we separate these rules for clarity of presentation. } Since the Neg rule for negation is quite similar to the Logical rule, we elide its description in the interest of space.

The last rule, Final, in Fig.~\ref{fig:build-hfta} simply assigns the final states of the constructed HFTA. Specifically, given  specification $\Phi$, we first use the auxiliary judgment $\leadsto$ to construct the HFTA $\hfta$ for $\Phi$. However, since we  want to accept only those programs that satisfy $\Phi$, we change the final states of the FTA that annotates the root node of $\hfta$ to be $\{q_{s_0}^\top\}$ rather than   $\{q_{s_0}^\top, q_{s_0}^\bot \}$.

\begin{figure}[!t]
\centering
\includegraphics[scale=0.57]{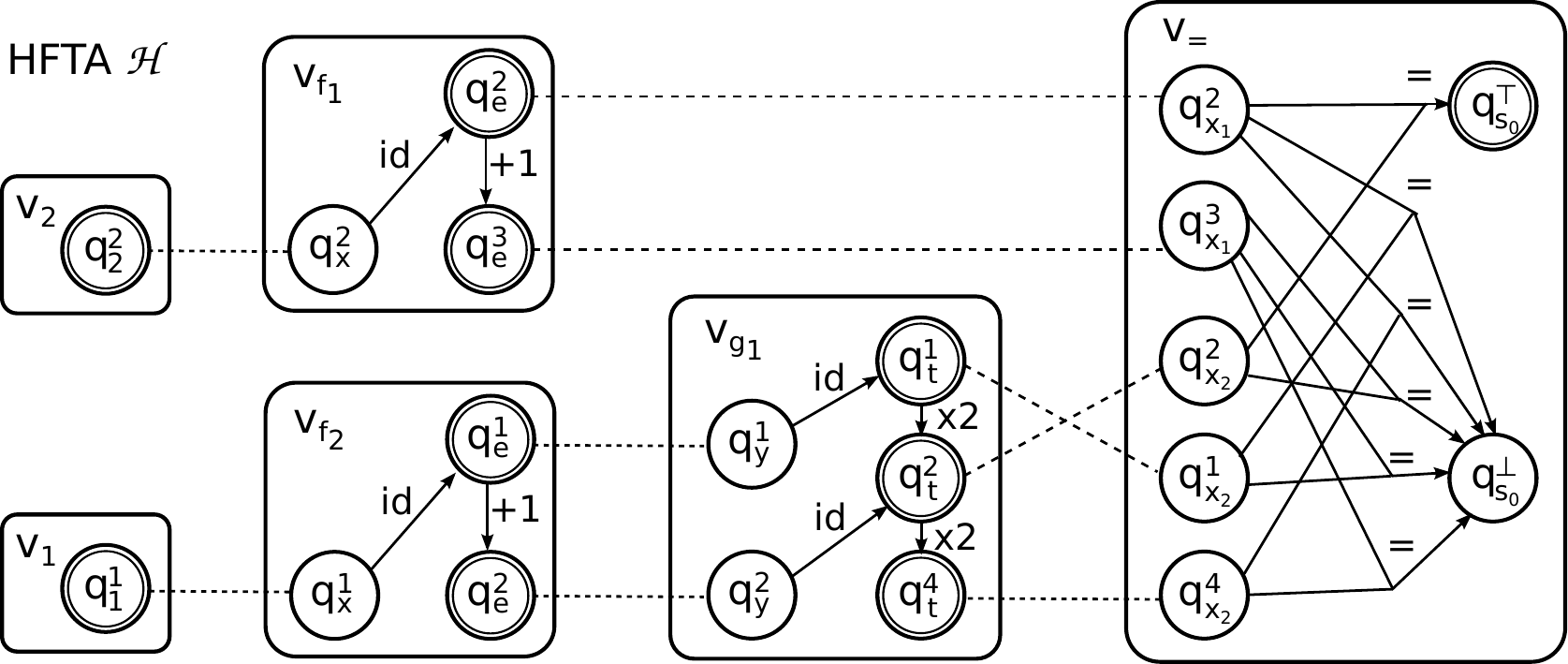}

\rule{0.75\textwidth}{0.5pt}

\vspace{0.05in}
\includegraphics[scale=0.57]{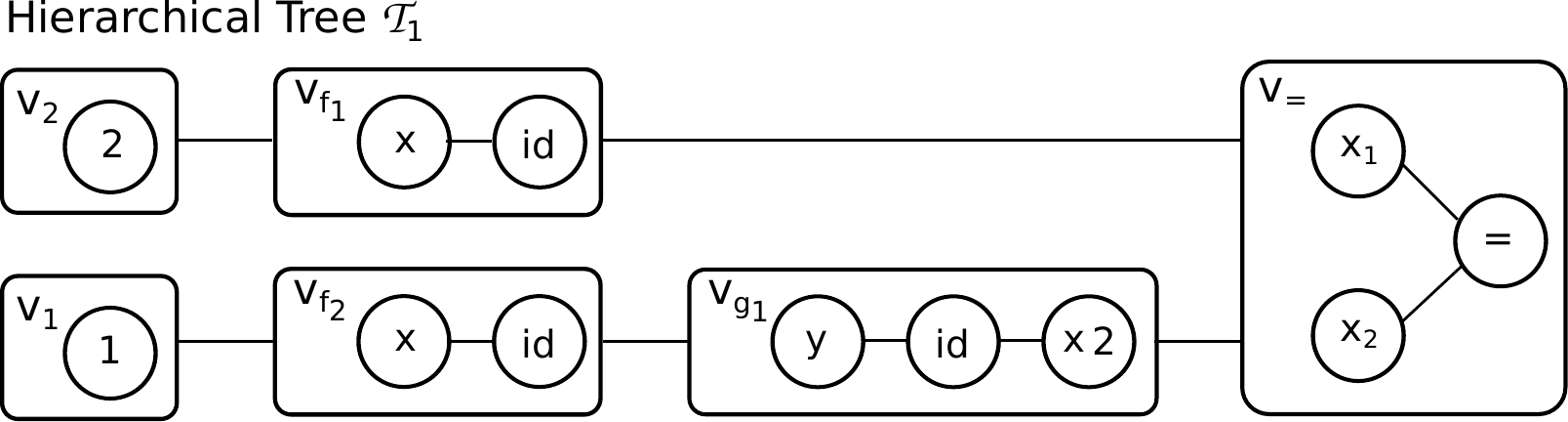}

\rule{0.75\textwidth}{0.5pt}

\vspace{0.05in}
\includegraphics[scale=0.57]{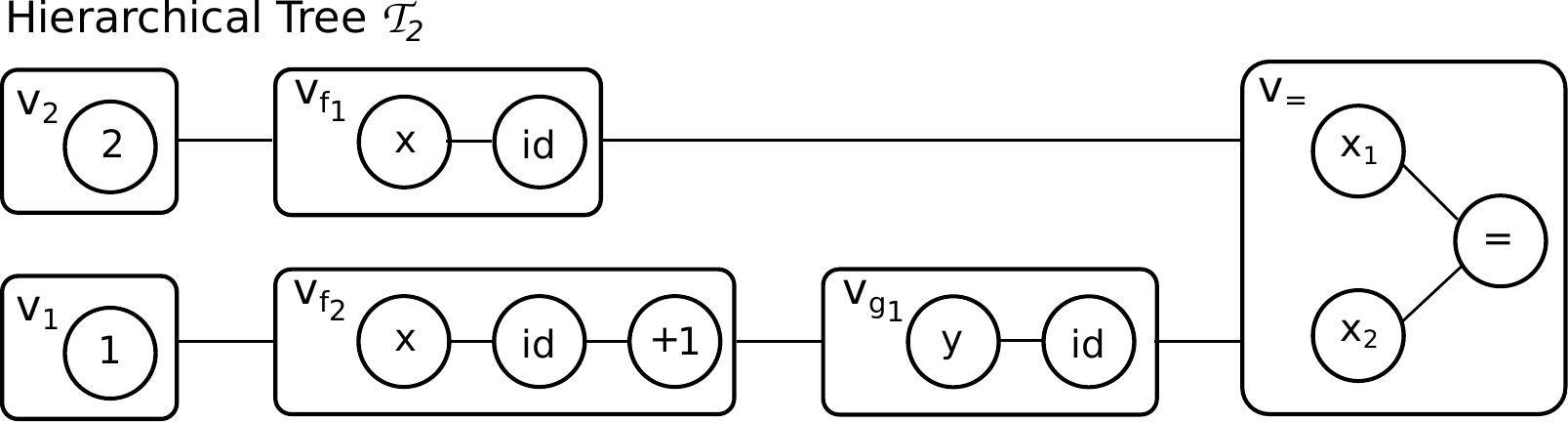}
\vspace{-5pt}
\caption{HFTA $\hfta = (\hftaNodes, \hftaAnnot, v_=, \hftaTrans)$ for specification $f_1(2) = g_1(f_2(1))$. FTAs are represented as hypergraphs, where circles correspond to  FTA states, double circles indicate final states, and labeled (hyper-)edges correspond to transitions. Inter-FTA transitions $\hftaTrans$ are represented as dashed lines. Hierarchical trees $\htree_1$ and $\htree_2$ are both accepted by~$\hfta$.}
\label{fig:hfta-plus-mult}
\vspace{-0.05in}
\end{figure}

\begin{example}\label{ex:hfta-plus-mult}{
Consider the  relational specification $\Phi: f(2) = g(f(1))$, where the DSL for $f$ is $e \rightarrow x ~|~ e + 1$ and the DSL for $g$ is $t \rightarrow y ~|~ t \times 2$ (Here, $x,y$ denote the inputs of $f$ and $g$ respectively). We now explain HFTA construction for this specification.
\begin{itemize}[leftmargin=*]
\item 
First, the relaxation procedure transforms the specification $\Phi$ to a relaxed version $\Phi': f_1(2) = g_1(f_2(1))$ and produces a mapping $\occurs = \big\{ f \mapsto \set{f_1, f_2}, ~ g \mapsto \set{g_1} \big\}$ relating the symbols in $\Phi, \Phi'$. 
\item 
Fig.~\ref{fig:hfta-plus-mult} shows the HFTA $\hfta$ constructed for $\Phi'$ where we view $+1$ and $\times 2$ as  unary functions for ease of illustration. 
Specifically, by the \textrm{Const} rule of Fig.~\ref{fig:build-hfta}, we build two nodes $v_1$ and $v_2$ that correspond to two FTAs that only accept constants 1 and 2, respectively. This sets up the initial state set of FTAs corresponding to $f_2$ and $f_1$, which results in two HFTA nodes $v_{f_2}$ and $v_{f_1}$ and their corresponding FTAs $\hftaAnnot(v_{f_2})$ and $\hftaAnnot(v_{f_1})$ constructed according to the \textrm{Func} rule. (Note that we build the FTAs using only two applications of the Prod rule.) Similarly, we build a node $v_{g_1}$ and FTA $\hftaAnnot(v_{g_1})$ by taking the final states of $\hftaAnnot(v_{f_2})$ as the initial state set of $\hftaAnnot(v_{g_1})$. Finally, the node $v_=$ along with its FTA $\hftaAnnot(v_=)$ is built by the \textrm{Logical} rule.
\item
Fig.~\ref{fig:hfta-plus-mult} also shows two hierarchical trees,  $\htree_1$ and  $\htree_2$, that are accepted by the HFTA $\hfta$ constructed above. However, note that only $\htree_1$ corresponds to a valid solution to the original synthesis problem because (1) in $\htree_1$ both $f_1$ and $f_2$ refer to the program $f = \lambda x. \ x$, and (2) in $\htree_2$ these two occurrences (i.e., $f_1$ and $f_2$) of $f$ correspond to different programs. 
\end{itemize}
}
\end{example}

\begin{theorem}{\bf (HFTA Soundness)}\label{thm:sound-relaxed}
Suppose $\grammars,\occurs \vdash \Phi \twoheadrightarrow \hfta$ according to the rules from Fig~\ref{fig:build-hfta}, and let $f_1, \ldots, f_n$ be the function symbols used in the relaxed specification $\Phi$. Given a hierarchical tree $\htree = (\htreeNodes, \htreeAnnot, \htreeRoot, \htreeEdges)$ that is accepted by $\hfta$, we have:
\begin{enumerate}
\item $\htreeAnnot(v_{f_i})$ is a program that conforms to grammar $\grammars(\occurs^{-1}(f_i))$
\item $\htreeAnnot(v_{f_1}), \ldots, \htreeAnnot(v_{f_n})$ satisfy $\Phi$, i.e. $\htreeAnnot(v_{f_1}), \ldots, \htreeAnnot(v_{f_n}) \models \Phi$
\end{enumerate}
\end{theorem}

\begin{proof}
See Appendix A.
\end{proof}

\begin{theorem}{\bf (HFTA Completeness)}\label{thm:complete-relaxed}
Let $\Phi$ be a ground formula where every function symbol $f_1, \ldots, f_n$ occurs exactly once, and suppose we have $\grammars,\occurs \vdash \Phi \twoheadrightarrow \hfta$.  If there are implementations $P_i$ of $f_i$ such that $P_1, \ldots, P_n \models \Phi$, where $P_i$ conforms to grammar  $\grammars(\occurs^{-1}(f_i))$, then there exists a hierarchical tree $\htree = (\htreeNodes, \htreeAnnot, \htreeRoot, \htreeEdges)$ accepted by $\hfta$ such that $\htreeAnnot(v_{f_i}) = P_i$ for all $i \in [1,n]$.
\end{theorem}

\begin{proof}
See Appendix A.
\end{proof}

\subsection{Enforcing Functional Consistency} \label{sec:find}

As stated by Theorems~\ref{thm:sound-relaxed} and~\ref{thm:complete-relaxed}, the HFTA method discussed in the previous subsection gives a sound and complete synthesis  procedure with respect to the \emph{relaxed} specification where each function symbol occurs exactly once. 
However, a hierarchical tree $\htree$ that is accepted by the HFTA might still violate functional consistency by assigning different programs to the same function $f$ used in the original  specification.
In this section, we describe an algorithm for finding hierarchical trees that  (a) are accepted by the HFTA and (b)  conform to the functional consistency requirement.

Algorithm~\ref{algo:find} describes our technique for finding accepting programs that obey functional consistency.  
Given an HFTA $\hfta$ and a mapping $\occurs$ from function symbols in the original specification to those in the relaxed specification, the {\sc FindProgs} procedure finds a mapping $\programs_{\text{res}}$ from each function symbol $f_i$ in the domain of $\occurs$ to a program $p_i$ such that $p_1, \ldots, p_n$ are accepted by $\hfta$. 
Since the resulting mapping $\programs_{\text{res}}$ maps each function symbol in the \emph{original} specification to a single program, the solution returned by {\sc FindProgs} is guaranteed to be a valid solution for the original relational synthesis problem.

We now explain how Algorithm~\ref{algo:find} works in more detail.
At a high level, {\sc FindProgs} is a recursive procedure that, in each invocation, updates the current mapping $\programs$ by finding a program for a currently unassigned function symbol $f$. In particular, we say that a function symbol $f$ is unassigned if $\programs$ maps $f$ to $\bot$. 
Initially, all function symbols are unassigned, but {\sc FindProgs} iteratively makes assignments to each function symbol in the domain of $\programs$.  
Eventually, if all function symbols have been assigned, this means that $\programs$ is a valid solution;
thus, Algorithm~\ref{algo:find} returns $\programs$ at line 3 if \textsf{ChooseUnassigned} yields \emph{null}.

\begin{algorithm}[t]
\caption{Algorithm for enforcing functional consistency.}
\label{algo:find}
\begin{algorithmic}[1]
\small 
\vspace{0.05in}
\Procedure{\textsc{FindProgs}}{$\hfta, \programs, \occurs$}
\vspace{0.05in}
\Statex \textbf{Input:} HFTA $\hfta = (\hftaNodes, \hftaAnnot, \hftaRoot, \hftaTrans)$, mapping $\programs$ that maps a function symbol to a program, and mapping $\occurs$ that maps a function symbol to its occurrences.
\Statex \textbf{Output:} Mapping $\programs_{\text{res}}$ that maps each function symbol to a program. 
\vspace{0.05in}
\State $f \gets \textsf{ChooseUnassigned}(\programs)$;
\If{$f = \emph{null}$} \Return $\programs$; \EndIf
\State $f_i \gets \textsf{ChooseOccurrence}(\occurs, f)$;
\ForEach{$p \in \big\{ \htreeAnnot(v_{f_i}) \ \big| \ (\htreeNodes, \htreeAnnot, \htreeRoot, \htreeEdges) \in \lang(\hfta), ~ v_{f_i} \in \htreeNodes \big\}$}
    \State $\hfta' \gets \textsc{Propagate}(\hfta, p, f, \occurs)$;
    \If{$\lang(\hfta') = \emptyset$} \textbf{continue}; \EndIf
    \State $\programs' \gets \programs[f \mapsto p]$;
    \State $\programs_{\text{res}} \gets \textsc{FindProgs}(\hfta', \programs', \occurs)$;
    \If{$\programs_{\text{res}} \neq \emph{null}$} \Return $\programs_{\text{res}}$; \EndIf
\EndFor
\State \Return $\emph{null}$;
\EndProcedure
\end{algorithmic}
\end{algorithm}

Given a  function symbol $f$ that is currently unassigned, {\sc FindProgs} first chooses some occurrence $f_i$ of $f$ in the relaxed specification (i.e., $f_i \in \occurs(f)$) and finds a program $p$ for $f_i$. 
In particular, given a hierarchical tree $\htree$ accepted by $\hfta$~\footnote{Section~\ref{sec:impl} describes a strategy for lazily enumerating hierarchical trees accepted by a given HFTA.}, we find the program $p$ that corresponds to $f_i$ in $\htree$ (line 5). Now, since every occurrence of $f$ must correspond to the same program, we use a procedure called {\sc Propagate} that propagates $p$ to all other occurrences of $f$ (line 6). The procedure {\sc Propagate} is shown in Algorithm~\ref{algo:aux} and returns a modified HFTA $\hfta'$ that enforces that all occurrences of $f$ are consistent with $p$. (We will discuss the  {\sc Propagate} procedure in more detail after finishing the discussion of {\sc FindProgs}).

The loop in lines 5--10 of Algorithm~\ref{algo:find} performs backtracking search. In particular, if the language of $\hfta'$ becomes empty after the call to {\sc Propagate}, this means that $p$ is not a suitable implementation of $f$ --- i.e., given the current mapping $\programs$, there is no extension of $\programs$ where  $p$ is assigned to $f$. Thus, the algorithm backtracks at line 7 by moving on to the next program for $f$. 

On the other hand, assuming that the call to {\sc Propagate} does not result in a contradiction (i.e., $\lang(\hfta') \neq \emptyset$), we try to find an implementation of the remaining function symbols via the recursive call to {\sc FindProgs} at line 9. 
If the recursive call does not yield \emph{null}, we have found a set of programs that both satisfy the relational specification and obey functional consistency; thus, we return $\programs_{\text{res}}$ at line 10. Otherwise, we again backtrack and look for a different implementation of $f$.

\paragraph{\bf \emph{Propagate subroutine.}} 
We now discuss the {\sc Propagate} subroutine for enforcing that different occurrences of a function have the same implementation. Given an HFTA $\hfta$ and a  function symbol $f$ with candidate implementation $p$, {\sc Propagate} returns a new HFTA $\hfta'$ such that the FTAs for all occurrences of $f$ only accept programs that have the same input-output behavior as $p$.

In more detail, the loop in lines 3--7 of Algorithm~\ref{algo:aux} iterates over all HFTA nodes $v$ that correspond to some occurrence of $f$. Since we want to make sure that the FTA for $v$ only accepts those programs that have the same input-output behavior as $p$, we first compute all final states $Q_f'$ that can be reached by successful runs of the FTA on program $p$ (line 5) and change the final states of $\hftaAnnot(v)$ to $Q_f'$ (line 6). Since some of the inter-FTA transitions become spurious as a result of this modification, we also remove inter-FTA transitions $q \rightarrow q'$ where $q$ is no longer a final state of its corresponding FTA (line 7). These modifications ensure that the FTAs for different occurrences of $f$ \emph{only} accept programs that have the same behavior as $p$.

\begin{figure}[!t]
\vspace{-10pt}
\begin{algorithm}[H]
\caption{Auxilary procedure for propagating a program in HFTA.}
\label{algo:aux}
\begin{algorithmic}[1]
\small 
\vspace{0.05in}
\Procedure{\textsc{Propagate}}{$\hfta, p, f, \occurs$}
\vspace{0.05in}
\Statex \textbf{Input:} HFTA $\hfta = (\hftaNodes, \hftaAnnot, \hftaRoot, \hftaTrans)$, program $p$, function symbol $f$, $\occurs$ maps function symbols to occurrences
\Statex \textbf{Output:} An updated HFTA $\hfta' = (\hftaNodes, \hftaAnnot', \hftaRoot, \hftaTrans')$
\vspace{0.05in}

\State $\hftaAnnot' \gets \hftaAnnot, \quad \hftaTrans' \gets \hftaTrans$;
\ForEach {$v \in \set{v_{f_i} ~|~ f_i \in \occurs(f), v_{f_i} \in \hftaNodes}$}
    \State $(\ftaStates, \ftaAlphabet, \ftaFinal, \ftaTrans) \gets \hftaAnnot(v)$;
    \State $\ftaFinal' \gets \textsf{ReachableFinalStates}(\hftaAnnot(v), p)$; \Comment{remove unreachable final states}
    \State $\hftaAnnot' \gets \hftaAnnot'[v \mapsto (\ftaStates, \ftaAlphabet, \ftaFinal', \ftaTrans)]$;
    \State $\hftaTrans' \gets \hftaTrans' \setminus \cup_{q \in \ftaFinal \setminus \ftaFinal'} \set{q \to q' ~|~ q \to q' \in \hftaTrans'}$; \Comment{remove spurious transitions}
\EndFor
\State \Return $(\hftaNodes, \hftaAnnot', \hftaRoot, \hftaTrans')$;
\EndProcedure
\end{algorithmic}
\end{algorithm}
\vspace{-10pt}
\end{figure}

\begin{figure}[!t]
\centering
\includegraphics[scale=0.6]{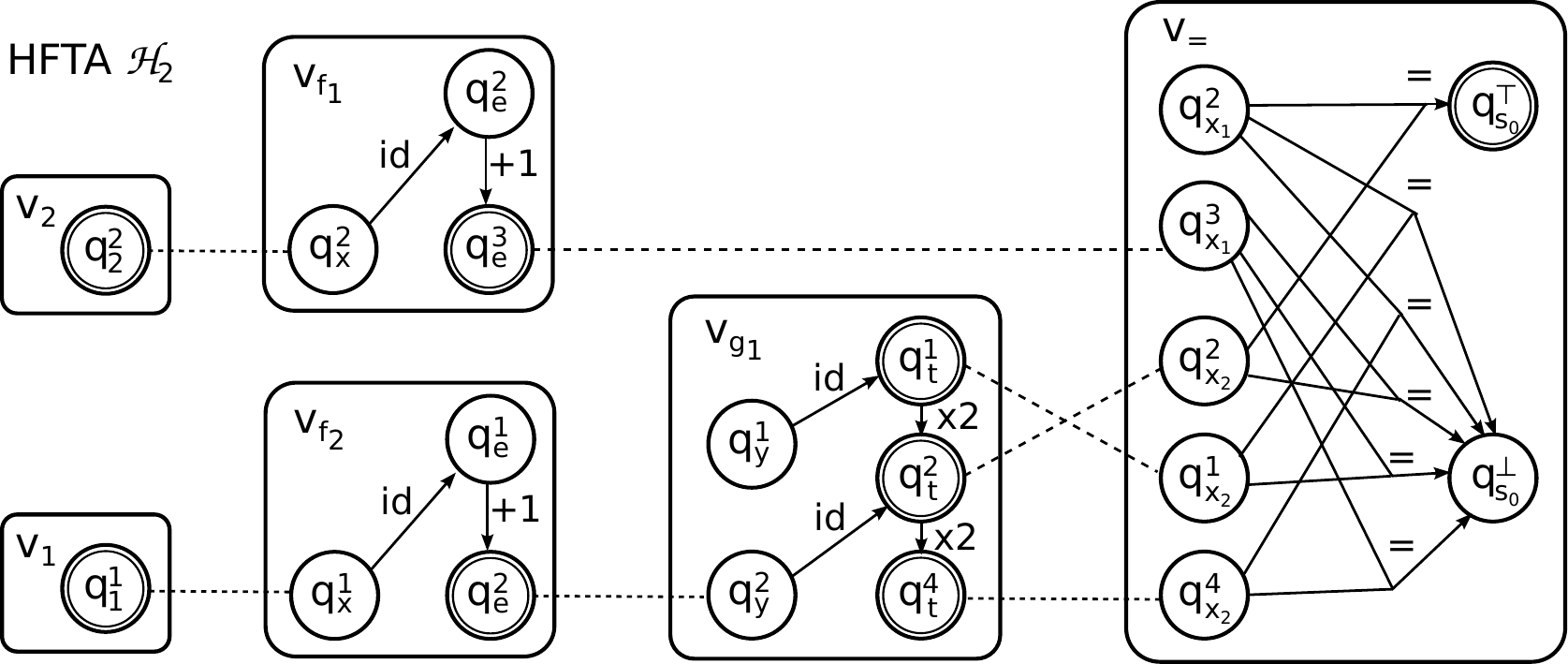}

\rule{0.75\textwidth}{0.5pt}

\vspace{0.05in}
\includegraphics[scale=0.6]{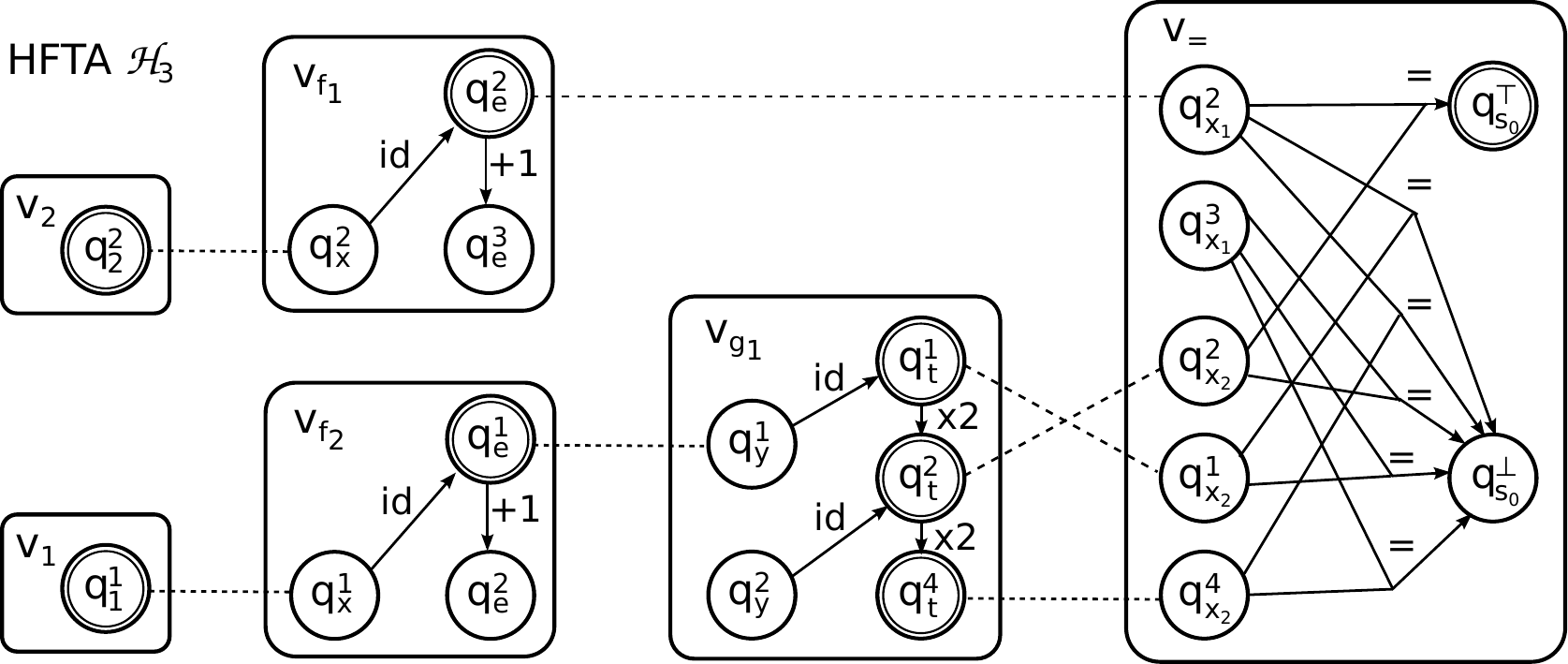}
\caption{HFTAs after propagation. $\hfta_2$ is obtained by propagating $f_1 = f_2 = \lambda x. ~ x+1$ on HFTA $\hfta$ in Example~\ref{ex:hfta-plus-mult} (shown in Fig.~\ref{fig:hfta-plus-mult}). $\hfta_3$ is obtained by propagating $f_1 = f_2 = \lambda x. ~ x$ on HFTA $\hfta$. }
\label{fig:hfta-propagated}
\vspace{-0.1in}
\end{figure}

\begin{example}
{
We now illustrate how {\sc FindProgs} extracts programs from the HFTA in Example~\ref{ex:hfta-plus-mult}. Suppose we first choose occurrence $f_2$ of function $f$ and  its implementation $f_2 = \lambda x. ~ x+1$. Since $f_1$ and $f_2$ must have the same implementation, the call to {\sc Propagate} results in the modified HFTA  $\hfta_2$ shown in Fig.~\ref{fig:hfta-propagated}. However, observe that $\lang(\hfta_2)$ is empty because the final state $q^\top_{s_0}$ in FTA $\hftaAnnot(v_=)$ is only reachable via state $q^2_{x_1}$, but $q^2_{x_1}$ no longer has incoming transitions. Thus, the algorithm backtracks to the only other implementation choice for $f_2$, namely $\lambda x. x$, and {\sc Propagate} now yields the HFTA $\hfta_3$ from Fig.~\ref{fig:hfta-propagated}. This time, $\lang(\hfta_3)$ is not empty; hence, the algorithm moves on to function symbol $g$. Since the only hierarchical tree in $\lang(\hfta_3)$ is $\htree_1$ from Fig.~\ref{fig:hfta-plus-mult}, we are forced to choose the implementation $g = \lambda y. ~ y \times 2$. Thus, the \textsc{FindProgs} procedure successfully returns $\programs = \set{f \mapsto \lambda x. x, ~ g \mapsto \lambda y. ~ y \times 2}$.
}
\end{example}

\subsection{Properties of the Synthesis Algorithm}

The following theorems state the soundness and completeness of the synthesis algorithm.

\begin{theorem}{\bf (Soundness)}\label{thm:sound}
Assuming the soundness of the \textsf{Verify} procedure invoked by Algorithm~\ref{algo:cegis}, then the set of programs $\programs$ returned by {\sc Synthesize} are guaranteed to satisfy the relational specification $\Psi$, and program $P_i$ of function $f_i$ conforms to grammar $\grammars(f_i)$.
\end{theorem}

\begin{proof}
See Appendix A.
\end{proof}

\begin{theorem}{\bf (Completeness)}\label{thm:complete}
If there is a set of programs that (a) satisfy $\Psi$ and (b) can be implemented using the DSLs given by $\grammars$, then Algorithm~\ref{algo:cegis} will eventually terminate with programs $\programs$ satisfying $\Psi$.
\end{theorem}

\begin{proof}
See Appendix A.
\end{proof}

\section{Implementation} \label{sec:impl}

We have implemented the proposed relational program synthesis approach in a framework called \toolname. To demonstrate the capabilities of the \toolname framework, we instantiate it on the first two domains from Section~\ref{sec:motivating-ex}, namely string encoders/decoders, and comparators. In this section, we describe our implementation of \toolname and its instantiations.

\subsection{Implementation of \toolname Framework}

While the implementation of the \toolname framework closely follows the algorithm described in Section~\ref{sec:synth}, it performs several important optimizations that we discuss next.

\vspace{-0.05in}
\paragraph{\textbf{Lazy enumeration.}} Recall that Algorithm~\ref{algo:find}  needs to (lazily) enumerate all hierarchical trees accepted by a given HFTA. Furthermore, since we want to find a program with the greatest generalization power, our algorithm should enumerate more promising programs first. Based on these criteria, we need a mechanism for predicting the generalization power of a hierarchical tree using some heuristic cost metric. In our implementation, we associate a non-negative cost with every DSL construct and compute the cost of a given hierarchical tree by summing up the costs of all its nodes. Because hierarchical trees with lower cost are likely to have better generalization power, our algorithm lazily enumerates hierarchical trees according to their cost.

In our implementation, we reduce the problem of enumerating hierarchical trees accepted by an HFTA to the task of enumerating  B-paths in a hypergraph~\cite{min-bpath}. In particular, we first flatten the HFTA into a standard FTA by combining the individual tree automata at each node via the inter-FTA transitions. We then represent the resulting flattened FTA as a hypergraph where the FTA states correspond to nodes and a transition $\sigma(q_1, \ldots, q_n) \to q$ corresponds to a B-edge $(\set{q_1, \ldots, q_n}, q)$ with weight $\emph{cost}(\sigma)$.
Given this representation, the problem of finding the lowest-cost hierarchical tree accepted by an HFTA becomes equivalent to the task of finding a minimum weighted B-path in a hypergraph, and our implementation leverages known algorithms for solving this problem~\cite{min-bpath}.

\vspace{-0.05in}
\paragraph{\textbf{Verification.}} Because our overall approach is based on the CEGIS paradigm,  we need a separate verification step to both check the correctness of the programs returned by the inductive synthesizer and  find counterexamples if necessary.
However, because heavy-weight verification can add considerable overhead to the CEGIS loop,  we \emph{test} the program against a large set of inputs rather than performing full-fledged verification in each iteration.  Specifically, we generate a set of validation inputs by computing all possible permutations of a finite set up to a bounded length $k$ and check the correctness of the candidate program against this validation set.  In our experiments, we use over a million test cases in each iteration, and resort to full verification only when the synthesized program passes all of these test cases.

\subsection{Instantiation for String Encoders and Decoders}\label{sec:encoder}

\begin{figure}[!t]
\centering

\begin{subfigure}{\textwidth}
\centering
\small
\[
\begin{array}{l l c l}
    \emph{EncodedText} & E &:=& M ~|~ \emph{padToMultiple}(E, num, char) ~|~ \emph{header}(E) \\
    \emph{Mapper} & M &:=& \emph{enc16}(B) ~|~ \emph{enc32}(B) ~|~ \emph{enc32Hex}(B) ~|~ \emph{enc64}(B) ~|~ \emph{enc64XML}(B) ~|~ \emph{encUU}(B) \\
    \emph{ByteArray} & B &:=& x ~|~ \emph{reshape}(B, num) ~|~ \emph{encUTF8}(I) ~|~ \emph{encUTF16}(I) ~|~ \emph{encUTF32}(I) \\
    \emph{IntArray} & I &:=& \emph{codePoint}(x) \\
\end{array}
\]
\[
x \in \emph{Variable} \qquad \qquad num \in \set{1, 2, \ldots, 8} \qquad \qquad char \in \set{\text{`=', `\'{ }',}\ldots}
\]

\vspace{-0.05in}
\caption{Context-free grammar for encoders.}
\label{fig:cfg-enc}
\end{subfigure}
\begin{subfigure}{\textwidth}
\centering
\small
\[
\begin{array}{l l c l}
    \emph{DecodedData} & D &:=& B ~|~ \emph{asUnicode}(I) \\
    \emph{IntArray} & I &:=& \emph{decUTF8}(B) ~|~ \emph{decUTF16}(B) ~|~ \emph{decUTF32}(B) \\
    \emph{ByteArray} & B &:=& M ~|~ \emph{invReshape}(B, num) \\
    \emph{Mapper} & M &:=& \emph{dec16}(C) ~|~ \emph{dec32}(C) ~|~ \emph{dec32Hex}(C) ~|~ \emph{dec64}(C) ~|~ \emph{dec64XML}(C) ~|~ \emph{decUU}(C) \\
    \emph{CharArray} & C &:=& x ~|~ \emph{removePad}(C, char) ~|~ \emph{substr}(C, num) \\
\end{array}
\]
\[
x \in \emph{Variable} \qquad \qquad num \in \set{1, 2, \ldots, 8} \qquad \qquad char \in \set{\text{`=', `\'{ }',}\ldots}
\]

\vspace{-0.05in}
\caption{Context-free grammar for decoders.}
\label{fig:cfg-dec}
\end{subfigure}

\vspace{-0.1in}
\caption{Context-free grammars for Unicode string encoders and decoders.}
\label{fig:cfg-codec}
\vspace{-0.2in}
\end{figure}

While \toolname is a generic framework that can be used in various application domains, one needs to construct suitable DSLs  and write relational specifications for each different domain. In this section, we discuss our instantiation of \toolname for synthesizing string encoders and decoders. Since we have already explained the relational property of interest in Section~\ref{sec:motivating-ex}, we discuss two simple DSLs, presented in Fig.~\ref{fig:cfg-codec}, for implementing encoders and decoders respectively.

\vspace{-0.05in}
\paragraph{\textbf{DSL for encoders}} We designed a DSL for string encoders by reviewing several different encoding mechanisms  and identifying their common building blocks. Specifically, this DSL allows transforming a Unicode string (or binary data) to a sequence of restricted ASCII characters (e.g., to fulfill various requirements of text-based network transmission protocols). At a high-level, programs in this DSL  first transform the input string  to an integer array and then to a byte array. The encoded text is obtained by applying various kinds of mappers to the byte array, padding it, and attaching length information. In what follows, we informally describe the semantics of the  constructs used in the encoder DSL.

\begin{itemize}[leftmargin=*]
\item {\bf \emph{Code point representation:}} The \emph{codePoint} function converts  string $s$ to an integer array $I$ where $I[j]$ corresponds to the Unicode code point for $s[j]$.
\item   {\bf \emph{Mappers:}} The \emph{encUTF8/16/32} functions transform the code point array to a byte array according to the corresponding standards of UTF-8, UTF-16, UTF-32, respectively. The other mappers \emph{encX} can further transform the byte array into a sequence of restricted ASCII characters based on different criteria. For example, the \emph{enc16} function is a simple hexadecimal mapper that can convert binary data \texttt{0x6E} to the string ``6E''.
\item  {\bf \emph{Padding:}} The \emph{padToMultiple} function takes an existing character array $E$ and pads it with a sequence of extra \emph{char} characters  to ensure that  the length of the padded sequence is evenly divisible by \emph{num}. 
\item  {\bf \emph{Header:}}  The \emph{header(E)} function prepends the ASCII representation of the length of the text to $E$.
\item  {\bf \emph{Reshaping:}}
The function \emph{reshape(B, num)} first concatenates all bytes in  $B$, then regroups them such that each group only contains \emph{num} bits (instead of 8), and finally generates a new byte array where each byte is equal to the value of corresponding group. For example, \emph{reshape}\texttt{([0xFF],4) = [0x0F,0x0F]} and \emph{reshape}\texttt{([0xFE],2) = [0x03,0x03,0x03,0x02]}.
\end{itemize}

\vspace{-0.05in}
\paragraph{\textbf{DSL for decoders}} The decoder DSL is quite similar to the encoder one and is structured as follows: Given an input string $x$, programs in this DSL first process $x$ by removing the header and/or padding characters and then transform it to a byte array using a set of pre-defined mappers. The decoded data can be either the resulting byte array or a Unicode string  obtained by converting the byte array to an integer Unicode code point array.
In more detail, the decoder DSL supports the following built-in operators:

\begin{itemize}[leftmargin=*]
\item \emph{\textbf{Unicode conversion:}} The \emph{asUnicode} function converts an integer array $I$ to a string $s$, where $s[j]$ corresponds to the Unicode symbol of code point $I[j]$.
\item \emph{\textbf{Mappers:}} The \emph{decUTF8/16/32} functions transform a byte array to a code point array based on standards of UTF-8, UTF-16, UTF-32, respectively. The other mappers \emph{decX} transform a sequence of ASCII characters to byte arrays according to their standard transformation rules.
\item \emph{\textbf{Character removal:}} The \emph{removePad} function removes all trailing characters \emph{char} from a given string $C$. The \emph{substr} function takes a string $C$ and an index \emph{num}, and returns the sub-string of $C$ from index \emph{num} to the end.
\item \emph{\textbf{Reshaping:}} The \emph{invReshape} function takes a byte array $B$, collects \emph{num} bits from the least significant end of each byte, concatenates the bit sequence and regroups  every eight bits to generate a new byte array. For instance, \emph{invReshape}\texttt{([0x0E,0x0F],4) = [0xEF]}.
\end{itemize}

\begin{example}
The desired {\tt encode/decode} functions from Example~\ref{ex:codec} can be implemented in our DSL as follows:
\[
\small 
\begin{BVerbatim}
encode(x) : padToMultiple (enc64 (reshape (encUTF8 (codePoint (x)), 6)), 4, '=')
decode(x) : asUnicode (decUTF8 (invReshape (dec64 (removePad (x, '=')), 6)))
\end{BVerbatim}
\]
\end{example}

\subsection{Instantiation for Comparators}

\begin{figure}[!t]
\centering
\small
\[
\begin{array}{llcl}
\emph{Comparator} & C & := & B \ | \ \emph{chain}(B, C) \\ 
\emph{Basic} & B & := & \emph{intCompare}(I_1, I_2) \ | \ \emph{strCompare}(S_1, S_2) \\ 
\emph{Integer} & I & := & \emph{countChar}(S, c) \ | \ \emph{length}(S) \ | \ \emph{toInt}(S) \\ 
\emph{String} & S & := & \emph{substr}(v, P_1, P_2) \\ 
\emph{Position} & P & := & \emph{pos}(v, t, k, d) \ | \ \emph{constPos}(k) \\ 
\end{array}
\]
\[
v \in \{ \emph{ x, y } \} \qquad  c \in \emph{Characters} \qquad   t \in \emph{Tokens}  \qquad  k \in \emph{Integers} \qquad  d \in \{ \emph{ Start, End } \}
\]

\vspace{-0.1in}
\caption{Context-free grammar for comparators.}
\label{fig:cfg-comparators}
\vspace{-0.2in}
\end{figure}

We have also instantiated the \toolname framework to enable automatic generation of custom string comparators. As described in Section~\ref{sec:motivating-ex}, this domain is an interesting ground for relational program synthesis because comparators must satisfy three different relational properties (i.e., anti-symmetry, transitivity, and totality). In what follows, we describe the DSL from Fig.~\ref{fig:cfg-comparators} that \toolname uses to synthesize these comparators. 

In more detail, programs in our comparator DSL take as  input a pair of strings \emph{x, y}, and return -1, 0, or 1 indicating that \emph{x}  precedes, is equal to, or succeeds  \emph{y} respectively. Specifically, a program is either a \emph{basic comparator} $B$ or a \emph{comparator chain} of the form $\emph{chain}(B_1, \ldots, B_n)$ which returns the result of the \emph{first} comparator that does not evaluate to zero. The DSL allows two basic comparators, namely  \emph{intCompare} and \emph{strCompare}. The integer or string inputs to these basic comparators can be obtained using the following functions:

\begin{itemize}[leftmargin=*]
\item  \emph{\textbf{Substring extraction:}} The \emph{substring} function is used to extract substrings of the input string. In particular, for a string $v$ and positions $P_1, P_2$, it returns  the substring of $v$ that starts at index $P_1$ and ends at index $P_2$.
\item  \emph{\textbf{Position identifiers:}} A position $P$ can either be a constant index ($\emph{constPos}(k)$) or the (start or end) index of the $k$'th occurrence of the match of token $t$ in the input string ($\emph{pos}(v, t, k, d)$).~\footnote{Tokens are chosen from a predefined universe of regular expressions.} For example, we have $\emph{pos}(\text{``} 12ab \text{''}, \texttt{Number}, 1, \texttt{Start}) = 0 $ and
$\emph{pos}(\text{``} 12ab \text{''}, \texttt{Number}, 1, \texttt{End}) = 2$ where \texttt{Number} is a token indicating a sequence of digits. 
\item  \emph{\textbf{Numeric string features:}} The DSL allows extracting various numeric features of a given string $S$. In particular, \emph{countChar} yields the number of occurrences of a given character $c$ in string $S$, \emph{length} yields string length, and \emph{toInt} converts a string representing an integer to an actual integer (i.e., $\emph{toInt}(``123") = 123$ but \emph{toInt}(``abc") throws an exception).
\end{itemize}

\begin{example}
Consider Example~\ref{ex:comparator} where the user wants to sort integers based on the number of occurrences of the number \texttt{5}, and, in the case of a tie,  sort them based on the actual integer values. 
This functionality can be implemented by the following simple program in our DSL:
\[\small 
\begin{BVerbatim}
chain (intCompare (countChar (x, '5'), countChar (y, '5')),
       intCompare (toInt (x), toInt (y)) )
\end{BVerbatim}
\]
\end{example}

\section{Evaluation}

We evaluate \toolname by using it to automatically synthesize (1) string encoders and decoders for program inversion tasks collected from prior work~\cite{inversion-loris} and  (2) string comparators to solve  sorting problems posted on StackOverflow.
The goal of our evaluation is to answer the following questions:
\begin{itemize}
\item How does \toolname perform on various relational synthesis tasks from two  domains?
\item What is the benefit of using HFTAs for relational program synthesis?
\end{itemize}

\paragraph{\textbf{Experimental setup.}} To evaluate the benefit of our approach over a base-line,
we compare our method against  \eusolver~\cite{eusolver}, an enumeration-based synthesizer that  won the General Track of the most recent SyGuS competition~\cite{sygus}. Since \eusolver only supports synthesis tasks in linear integer arithmetic, bitvectors, and basic string manipulations by default, we extend it to encoder/decoders and comparators by implementing the same DSLs  described in Section~\ref{sec:impl}. Additionally,  we implement the same CEGIS loop used in \toolname for \eusolver and use the same bounded verifier in our evaluation. 
All experiments are conducted on a machine with Intel Xeon(R) E5-1620 v3 CPU and 32GB of physical memory, running the Ubuntu 14.04 operating system. Due to finite computational resources, we set a time limit of 24 hours for each benchmark.

\subsection{Results for String Encoders and Decoders}

In our first experiment, we evaluate \toolname by using it to simultaneously synthesize Unicode string encoders and decoders, which are required to be inverses of each other.

\paragraph{\textbf{Benchmarks.}} We collect a total of ten  encoder/decoder benchmarks, seven of which are taken from a prior paper on program inversion~\cite{inversion-loris}.
Since the 14 benchmarks from prior paper~\cite{inversion-loris} are essentially seven pairs of encoders and decoders, we have covered all their encoder/decoder benchmarks.
The remaining three benchmarks, namely, \emph{Base32hex, UTF-32, and UTF-7}, are also well-known encodings. Unlike previous work on program inversion, our goal is to solve the considerably more difficult problem of simultaneously synthesizing the encoder \emph{and}  decoder from input-output examples rather than inverting an existing function. For each benchmark, we use 2-3 input-output examples taken from the documentation of the corresponding encoders. We also specify the relational property $\forall x.$ \texttt{decode(encode(x))=x} and use the encoder/decoder DSLs presented in Section~\ref{sec:encoder}.

\paragraph{\textbf{Main results.}}
Our main experimental results are summarized in Table~\ref{tab:result-codec}, where the first two columns (namely, ``Enc Size'' and ``Dec Size'') describe the size of the target program in terms of the number of AST nodes.~\footnote{We obtain this information by manually writing a simplest DSL program for achieving the desired task.} The next three columns under {\bf \toolname} summarize the results obtained by running \toolname on each of these benchmarks, and the three columns under {\sc \bf \eusolver} report the same results for \eusolver. Specifically, the column labeled ``Iters'' shows the total number of iterations inside the CEGIS loop,  ``Total'' shows the total synthesis time in seconds, and ``Synth'' indicates the time (in seconds) taken by the inductive synthesizer (i.e., excluding verification). If a tool fails to solve the desired task within the 24 hour time limit, we write T/O to indicate a time-out. Finally, the last column labeled ``Speed-up'' shows the speed-up of \toolname over \eusolver for those benchmarks where neither tool times out. 

\begin{table}[t]
\centering
\footnotesize
\caption{Experimental results on Unicode string encoders and decoders.}
\label{tab:result-codec}
\vspace{-0.1in}
\begin{tabular}{|c|c|c|c|c|c|c|c|c|c|c|}
\hline
\multirow{2}{*}{\textbf{Benchmark}} & \textbf{Enc} & \textbf{Dec} & \multicolumn{4}{c|}{\textbf{\toolname}} & \multicolumn{3}{c|}{\textbf{\eusolver}} & \multirow{2}{*}{\textbf{Speedup}} \\
\cline{4-10}
& \textbf{Size} & \textbf{Size} & \textbf{Iters} & \textbf{Total(s)} & \textbf{Synth(s)} & \!\textbf{Mem(MB)}\! & \textbf{Iters} & \textbf{Total(s)} & \textbf{Synth(s)} & \\
\hline
Base16 & 6 & 6 & 3 & 16.2 & 10.4 & 551 & 3 & 494.2 & 489.3 & 30.4x \\
\hline
Base32 & 9 & 8 & 5 & 21.0 & 14.6 & 458 & - & T/O & T/O & - \\
\hline
Base32hex & 9 & 8 & 5 & 22.9 & 15.6 & 468 & - & T/O & T/O & - \\
\hline
Base64 & 9 & 8 & 5 & 16.4 & 8.8 & 916 & - & T/O & T/O & - \\
\hline
Base64xml & 9 & 8 & 6 & 23.4 & 15.5 & 1843 & - & T/O & T/O & - \\
\hline
UU & 10 & 10 & 5 & 23.3 & 15.2 & 843 & - & T/O & T/O & - \\
\hline
UTF-8 & 6 & 6 & 4 & 17.7 & 11.9 & 916 & 4 & 536.4 & 531.8 & 30.2x \\
\hline
UTF-16 & 6 & 6 & 2 & 11.5 & 5.4 & 285 & 4 & 1265.2 & 1259.9 & 110.4x \\
\hline
UTF-32 & 6 & 6 & 2 & 11.6 & 4.8 & 284 & 3 & 771.1 & 765.1 & 66.4x \\
\hline
UTF-7 & 6 & 6 & 3 & 14.8 & 9.2 & 285 & 3 & 475.2 & 470.1 & 32.1x \\
\hline
\hline
\textbf{Average} & 7.6 & 7.2 & 4.0 & 17.9 & 11.1 & 685 & 3.4 & 708.4 & 703.2 & 46.5x \\
\hline
\end{tabular}
\vspace{-0.2in}
\end{table}

As shown in Table~\ref{tab:result-codec}, \toolname can  correctly solve all of these benchmarks~\footnote{We manually inspected the synthesized solutions and  confirmed their correctness.} and takes an average of 17.9 seconds per benchmark. In contrast, \eusolver solves half of the benchmarks within the 24 hour time limit and takes an average of approximately 12 \emph{minutes} per benchmark that it is able to solve. For the five benchmarks that can be solved by  both tools, the average speed-up of \toolname over \eusolver is 46.5x.~\footnote{Here, we use geometric mean to compute the average since the arithmetic mean is not meaningful for ratios.} These statistics clearly demonstrate the advantages of our HFTA-based approach compared to enumerative search: Even though both tools use the same DSLs, verifier, and CEGIS architecture, \toolname unequivocally outperforms \eusolver across all benchmarks.

Next, we compare \toolname and \eusolver in terms of the number of CEGIS iterations for the five benchmarks that can be solved by both tools. As we can see from Table~\ref{tab:result-codec},  \toolname  takes 2.8 CEGIS iterations on average, whereas \eusolver needs an average of 3.4 attempts to find the correct program. This discrepancy suggests that our HFTA-based method might have better generalization power compared to enumerative search. In particular, our method first generates a version space that contains \emph{all} tuples of programs that satisfy the relational specification and then searches for the best program in this version space. In contrast, \eusolver returns the \emph{first} program that satisfies the current set of counterexamples; however, this program may not be the best (i.e., lowest-cost) one in \toolname's version space.

Finally, we compare \toolname against \eusolver in terms of synthesis time per CEGIS iteration: \toolname takes an average of 4.5 seconds per iteration, whereas \eusolver takes 208.4 seconds on average across the five benchmarks that it is able to solve.  To summarize, these results clearly indicate the advantages of our approach compared to enumerative search when synthesizing string encoders and decoders.

\paragraph{\textbf{Memory usage.}}
We now investigate the memory usage of \toolname on the encoder/decoder benchmarks (see column labeled Mem in Table~\ref{tab:result-codec}).  Here, the memory usage varies between 284MB and 1843MB, with the average being 685MB. As we can see from Table~\ref{tab:result-codec}, 
the memory usage mainly depends on (1) the number of CEGIS iterations and (2) the size of programs to be synthesized.
Specifically, as the CEGIS loop progresses, the number of function occurrences in the ground relational specification increases, which results in larger HFTAs. For example, the impact of the number of CEGIS iterations on memory usage becomes apparent by comparing the  ``Base64xml'' and ``UTF-32'' benchmarks, which take 6 and 2 iterations and consume 1843 MB and 284 MB respectively. In addition to the number of CEGIS iterations, the size of the target program also has an impact on memory usage. Intuitively, the larger the synthesized programs, the larger the size of the individual FTAs; thus, memory usage tends to increase with program size.  For example, the impact of program size on memory usage is illustrated by the difference between the ``UU'' and ``Base32'' benchmarks.

\subsection{Results for String Comparators}

In our second experiment, we evaluate \toolname by using it to synthesize string comparators for  sorting problems obtained from StackOverflow.  Even though our goal is to synthesize a single {\tt compare} function, this problem is still a relational synthesis task because the generated program must obey two 2-safety properties (i.e., anti-symmetry and totality) and one 3-safety property (i.e., transitivity). Thus, we believe that comparator synthesis is also an interesting and relevant test-bed for evaluating relational synthesizers.

\paragraph{\textbf{Benchmarks.}}

To perform our evaluation, we collected 20 benchmarks from StackOverflow using the following methodology: 
First, we searched StackOverflow for the keywords ``\emph{Java string comparator}''.
Then, we manually inspected each of these posts and retained exactly those that satisfy the following criteria: 
\begin{itemize}
\item 
The question in the post should involve writing a \texttt{compare} function for sorting strings. 
\item 
The post should contain a list of sample strings that are sorted in the desired way. 
\item 
The post should contain a natural language description of the desired sorting task. 
\end{itemize}

The relational specification $\Psi$ for each benchmark consists of the following three parts: 
\begin{itemize}
\item 
Universally-quantified formulas reflecting the three relational properties that \texttt{compare} has to satisfy (i.e., anti-symmetry, transitivity, and totality). 
\item 
Another quantified formula  that stipulates  reflexivity (i.e., $\forall x. \ ${\tt compare}($x,x$) = 0)
\item 
Quantifier-free formulas that correspond to the input-output examples from the StackOverflow post. In particular, given a sorted  list $l$, if string $x$ appears before string $y$ in $l$, we add the examples \texttt{compare(x,y)} = \texttt{-1} and \texttt{compare(y,x)} = \texttt{1}.
\end{itemize}

Among these benchmarks, the number of examples range from 2 to 30, with an average of 16.

\paragraph{\textbf{Main results.}}

\begin{table}[t]
\centering
\footnotesize 
\caption{Experimental results on comparators.}
\label{tab:result-comp}
\vspace{-0.1in}
\begin{tabular}{|c|c|c|c|c|c|c|c|c|c|}
\hline
\multirow{2}{*}{\textbf{Benchmark}} & \multirow{2}{*}{\textbf{Size}} & \multicolumn{4}{c|}{\textbf{\toolname}} & \multicolumn{3}{c|}{\textbf{\eusolver}} & \multirow{2}{*}{\textbf{Speedup}} \\
\cline{3-9}
& & \textbf{Iters} & \textbf{Total(s)} & \textbf{Synth(s)} & \textbf{Mem(MB)} & \textbf{Iters} & \textbf{Total(s)} & \textbf{Synth(s)} & \\
\hline
comparator-1 & 5 & 2 & 2.0 & 0.1  &  12 &     2 & 2.6 & 0.7 & 1.3x \\
\hline
comparator-2 & 13 & 5 & 3.4 & 0.2 &   1375 &        5 & 5.5 & 2.5 & 1.6x \\
\hline
comparator-3 & 13 & 5 & 4.1 & 0.5 &  385 &         6 & 21.5 & 18.5 & 5.3x \\
\hline
comparator-4 & 23 & 6 & 7.7 & 1.4 &  1401 &         10 & 299.7 & 291.9 & 38.8x \\
\hline
comparator-5 & 21 & 7 & 8.8 & 1.2 &  548 &           10 & 650.9 & 643.3 & 74.1x \\
\hline
comparator-6 & 23 & 4 & 9.1 & 0.2 &   30  &         11 & 57.3 & 52.5 & 6.3x \\
\hline
comparator-7 & 25 & 7 & 9.8 & 0.6 &   1386 &            9 & 640.8 & 632.5 & 65.1x \\
\hline
comparator-8 & 23 & 6 & 9.9 & 1.6 &    117 &         7 & 50.7 & 45.2 & 5.1x \\
\hline
comparator-9 & 41 & 10 & 25.5 & 12.7  &    2146             & - & T/O & T/O & - \\
\hline
comparator-10 & 21 & 11 & 40.7 & 32.4  &  1454             & 12 & 1301.9 & 1295.3 & 32.0x \\
\hline
comparator-11 & 41 & 9 & 42.9 & 7.4  & 4507            & 13 & 13966.0 & 13952.7 & 325.2x \\
\hline
comparator-12 & 17 & 10 & 75.2 & 66.7  & 1045              & 15 & 9171.0 & 9161.4 & 122.0x \\
\hline
comparator-13 & 27 & 10 & 85.0 & 60.1    & 3834              & 13 & 10742.8 & 10736.3 & 126.3x \\
\hline
comparator-14 & 25 & 8 & 102.5 & 93.2   & 2022            & - & T/O & T/O & - \\
\hline
comparator-15 & 17 & 9 & 116.7 & 112.3   &  1622                & 16 & 26443.4 & 26434.6 & 226.5x \\
\hline
comparator-16 & 19 & 7 & 130.1 & 124.2    &   1383                   & - & T/O & T/O & - \\
\hline
comparator-17 & 43 & 11 & 196.2 & 183.4   &  4460                 & - & T/O & T/O & - \\
\hline
comparator-18 & 41 & 10 & 406.5 & 351.0   & 18184            & - & T/O & T/O & - \\
\hline
comparator-19 & 65 & 9 & 523.9 & 485.7  & 18367                 & - & T/O & T/O & - \\
\hline
comparator-20 & - & - & T/O & T/O        & -             & - & T/O & T/O & - \\
\hline
\hline
\textbf{Average} & 26.5 & 7.7 & 94.7 & 80.8 &  3382   & 9.9 & 4873.4 & 4866.7 & 26.0x \\
\hline
\end{tabular}
\vspace{-0.2in}
\end{table}

Our main results are summarized in Table \ref{tab:result-comp}, which is structured in the same way as Table~\ref{tab:result-codec}.
The main take-away message from this experiment is that \toolname can successfully solve 95\% of the benchmarks within the 24 hour time limit. Among these 19 benchmarks,  \toolname takes an average of 94.7 seconds per benchmark, and it solves  55\% of the benchmarks within 1 minute and 75\% of the benchmarks within 2 minutes.

In contrast to \toolname, \eusolver solves considerably fewer benchmarks within the 24 hour time-limit. In particular, \toolname solves 46\% more benchmarks than \eusolver (19 vs. 13) and outperforms \eusolver by 26x (in terms of running time) on the common benchmarks that can be solved by both techniques. Furthermore, similar to the previous experiment, we also observe that \toolname requires fewer CEGIS iterations (7.0 vs. 9.9 on average), again confirming the hypothesis that the HFTA-based approach might have better generalization power. Finally, we note that \toolname is also more efficient than \eusolver per CEGIS iteration (12 seconds vs. 492 seconds).

\paragraph{\textbf{Memory usage.}}
We also investigate the memory usage of \toolname on the comparator benchmarks. As shown in the Mem column of Table~\ref{tab:result-comp}, memory usage varies between 12 MB and 18367MB, with  an average memory consumption of  3382MB. Comparing these statistics with Table~\ref{tab:result-codec}, we see that memory usage is higher for comparators than the encoder/decoder benchmarks. We believe this difference can be attributed to the following three factors:  First, as shown in Table~\ref{tab:result-comp}, the size of the synthesized programs is larger for the comparator domain. Second, the relational specification for comparators is more complex and involves multiple properties such as reflexivity, anti-symmetry, totality, and transitivity. Finally, most benchmarks in the comparator domain require more CEGIS iterations to solve and therefore result in larger HFTAs.

\paragraph{\textbf{Analysis of failed benchmarks.}} We manually inspected the benchmark ``comparator-20'' that \toolname failed to synthesize within the 24 hour time limit. In particular, we found this benchmark is not expressible in our current DSL because it requires comparing integers that are obtained by concatenating all substrings that represent integers in the input strings.

\paragraph{\textbf{Summary.}} In summary, this experiment demonstrates that \toolname can successfully synthesize non-trivial string comparators that arise in real-world scenarios. This experiment also demonstrates the advantages of our new relational synthesis approach compared to an existing state-of-the-art solver. While the comparator synthesis task involves synthesizing a \emph{single} function, the enumeration-based approach performs considerably worse than \toolname because it does not use the relational (i.e., $k$-safety) specification to prune its search space.

\section{Related Work}

In this section, we survey prior work that is most closely related to relational program synthesis.

\vspace{-0.05in}
\paragraph{\textbf{Relational Program Verification.}} There is a large body of work on \emph{verifying} relational properties about programs~\cite{rhl,chl,relsep,product1,mediator,merge}. Existing work on relational verification can be generally categorized into three classes, namely \emph{relational program logics}, \emph{product programs}, and \emph{Constrained Horn Clause (CHC) solving}. The first category includes Benton's Relational Hoare Logic~\cite{rhl} and its variants~\cite{prhl,prhl2,relsep} as well as Cartesian Hoare Logic~\cite{chl,qchl} for verifying $k$-safety properties. In contrast to these approaches that provide a dedicated program logic for reasoning about relational properties, an alternative approach is to build a so-called \emph{product program} that captures the simultaneous execution of multiple programs or different runs of the same program~\cite{product1,barthe2004secure,product-asymmetric}. In the simplest case, this approach sequentially composes different programs (or copies of the same program)~\cite{barthe2004secure}, but more sophisticated product construction methods perform various transformations such as loop fusion and unrolling to make invariant generation easier~\cite{product1, product-asymmetric,program-consolidation,covac}. A common theme underlying all these product construction techniques is to reduce the relational verification problem to a standard safety checking problem.  Another alternative approach that has been explored in prior work is to reduce the relational verification problem to solving a (recursive) set of Constrained Horn Clauses (CHC) and apply transformations that make the problem easier to solve~\cite{chc1,chc2}. To the best of our knowledge, this paper is the first one to address the dual \emph{synthesis} variant of the relational verification problem.

\vspace{-0.05in}
\paragraph{\textbf{Program Synthesis.}} This paper is related to a long line of work on program synthesis dating back to the 1960s~\cite{green}.  Generally speaking, program synthesis techniques can be classified into two classes depending on whether they perform \emph{deductive} or \emph{inductive} reasoning. In particular, deductive synthesizers generate correct-by-construction programs by applying refinement and transformation rules~\cite{refinement,manna1986deductive,fiat,leon,synquid}.
In contrast, inductive synthesizers learn programs from input-output examples using techniques such as constraint solving~\cite{sketch1,simpl}, enumerative search~\cite{eusolver,lambda2}, version space learning~\cite{prose}, stochastic search~\cite{stratified-synth,stoke}, and statistical models and machine learning~\cite{deepsyn,jsnice}.
Similar to most recent work in this area~\cite{sketch1,prose,syngar,gulwani2012synthesis}, our proposed method also uses inductive synthesis. However, a key difference is that we use relational examples in the form of ground formulas rather than the more standard input-output examples.

\vspace{-0.05in}
\paragraph{\textbf{Counterexample-guided Inductive Synthesis.}} The method proposed in this paper follows the popular counterexample-guided inductive synthesis (CEGIS) methodology~\cite{sketch1,sketch2,sygus}. In the CEGIS paradigm, an inductive synthesizer generates a candidate program $P$ that is consistent with a set of examples, and the verifier checks the correctness of $P$ and provides  counterexamples when $P$ does not meet the user-provided specification. Compared to other synthesis algorithms that follow the CEGIS paradigm, the key differences of our method are (a) the use of relational counterexamples and (b) a new inductive synthesis algorithm that utilizes relational specifications.

\vspace{-0.05in}
\paragraph{\textbf{Version Space Learning.}} As mentioned earlier, the inductive synthesizer used in this work can be viewed as a generalization of \emph{version space learning} to the relational setting. The notion of \emph{version space} 
was originally introduced in the 1980s as a supervised learning framework~\cite{vs} and has found numerous applications within the field of program synthesis~\cite{vsa,prose,flashfill,dace,fidex,mitra}. Generally speaking, synthesis algorithms based on version space learning construct some sort of data structure that represents all programs that are consistent with the examples. While existing version space learning algorithms only work with input-output examples, the method proposed here works with arbitrary ground formulas representing relational counterexamples and uses hierarchical finite tree automata to compose the version spaces of individual functions.

\vspace{-0.05in}
\paragraph{\textbf{Program Inversion.}} Prior work has addressed the \emph{program inversion} problem, where the goal is to automatically generate the inverse of a given program~\cite{inv1,inv2,inv3,inversion-swarat,inversion-loris}. Among these, the PINS tool uses inductive synthesis to semi-automate the inversion process based on templates provided by the user~\cite{inversion-swarat}. More recent work describes a fully automated technique, based on symbolic transducers, to generate the inverse of a given program~\cite{inversion-loris}. While program inversion is one of the applications that we consider in this paper, relational synthesis is applicable to many problems beyond program inversion. Furthermore, our approach can be used to  synthesize the program and its inverse \emph{simultaneously} rather than requiring the user to provide one of these functions.

\section{Limitations}

While we have successfully used the proposed relational synthesis method to synthesize encoder-decoder pairs and comparators, our current approach has some limitations that we plan to address in future work. First, our method only works with simple DSLs without recursion, loops, or let bindings.  That is, the programs that can be currently synthesized by \toolname are compositions of built-in functions provided by the DSL. Second, we only allow relational specifications of the form $\forall \vec{x}. \phi(\vec{x})$ where $\phi$ is quantifier-free. Thus, our method does not handle more complex relational specifications with quantifier alternation (e.g., $\forall x. \exists y. f(x,y) = g(y, x)$).

\section{Conclusion}

In this paper, we have introduced \emph{relational program synthesis} ---the problem of synthesizing one or more programs that collectively satisfy a relational specification--- and described its numerous applications in real-world programming scenarios. We have also taken a first step towards solving this problem by presenting a CEGIS-based approach with a novel inductive synthesis component. In particular, the key idea is to construct a relational version space (in the form of a hierarchical finite tree automaton) that encodes all tuples of programs that satisfy  the original specification.

We have implemented this technique in a relational synthesis framework called \toolname which can be instantiated in different application domains by providing a suitable domain-specific language as well as the relevant relational specifications. Our evaluation in two different application domains (namely, encoders/decoders and comparators) demonstrate that \toolname can effectively synthesize  pairs of closely related programs (i.e., inverses) as well as individual programs  that must obey non-trivial $k$-safety specifications (e.g., transitivity). Our evaluation also shows that \toolname significantly outperforms \eusolver, both in terms of the efficiency of inductive synthesis as well as its generalization power per CEGIS iteration.

\section*{Acknowledgments}\label{sec:ack}
We would like to thank Yu Feng, Kostas Ferles, Jiayi Wei, Greg Anderson, and the anonymous OOPSLA'18 reviewers for their thorough and helpful comments on an earlier version of this paper. This material is based on research sponsored by NSF Awards \#1712067, \#1811865, \#1646522, and AFRL Award \#8750-14-2-0270. The views and conclusions contained herein are those of the authors and should not be interpreted as necessarily representing the official policies or endorsements, either expressed or implied, of the U.S. Government.

\bibliography{main}

\newcommand*{\extended}{} 
\ifdefined\extended
\appendix
\section{Proof of Theorems}


\begin{lemma}\label{lem:sound-fta}
Suppose FTA $\fta = (\ftaStates, \ftaAlphabet, \ftaFinal, \ftaTrans)$ is built with grammar $G$ and initial states $[\ftaStates_1, \ldots, \ftaStates_m]$ using rules from Figure~\ref{fig:buildFTA}. Given a tree $t(x_1, \ldots, x_m)$ that is accepted by $\fta$, we have
\begin{enumerate}
\item $t$ conforms to grammar $G$.
\item The accepting run of $t$ with $x_i = c_i$, where $q^{c_i}_{x_i} \in \ftaStates_i$, is rooted with final state $q^c_s$ and $c = \denot{t(c_1, \ldots, c_m)}$.
\end{enumerate}
\end{lemma}

\begin{proof}
To prove by induction on height of $t$, we first strengthen (2) to: the run of $t$ with $x_i = c_i$ is rooted with state $q^c_s$ and $c = \denot{t(c_1, \ldots, c_m)}$.
\begin{itemize}
\item Base case: height is 1.
In this case, $t$ only has one variable $x$, and $q^c_{s}$ is passed as an input.
Consider the \textsf{Input} rule in Figure~\ref{fig:buildFTA}, we know $q^c_{x} \in \ftaStates$ and $x \to q^c_x \in \ftaTrans$.
Given $t$ is accepted by $\fta$, $q^c_x$ is a final state, so $x$ conforms to grammar $G$.
The run has only one node $q^c_x$, and $c = \denot{t(c)}$.
\item Inductive case: Suppose the lemma is correct for term with height no larger than $h$, and the goal is to prove it is correct for term $t$ of height $h+1$.
Without loss of generality, assume $t = \sigma(t_1, \ldots, t_n)$ is of height $h+1$, so the height of terms $t_1, \ldots, t_n$ is at most $h$.
By inductive hypothesis, $t_i$ conforms to grammar $G$ for $i \in [1, n]$.
Furthermore, the run of $t_i$ is rooted with state $q^{a_i}_{s_i}$ where $a_i = \denot{t_i}$ on environment $\set{x_j \mapsto c_j ~|~ j \in [1, m]}$. 
Based on the \textsf{Prod} rule in Figure~\ref{fig:buildFTA}, we know the transition $\sigma(q^{a_1}_{s_1}, \ldots, q^{a_n}_{s_n}) \to q^c_s$ is only added to $\ftaTrans$ if there is production $\sigma(s_1, \ldots, s_n) \to s$ in grammar $G$ and $c = \denot{\sigma(a_1, \ldots, a_n)}$.
Therefore, $t$ conforms to grammar $G$ and the run of $t$ with $x_i = c_i$ is rooted with state $q^c_s$ and $c = \denot{t(c_1, \ldots, c_m)}$.
\end{itemize}
Since term $t$ is accepted by $\fta$, we know $q^c_s$ is a final state, which concludes the proof.
\end{proof}

\begin{lemma}\label{lem:sound-relaxed-term}
Suppose $\grammars,\occurs \vdash t \leadsto \hfta$ according to the rules from Figure~\ref{fig:build-hfta}, where $t$ is a term, and let $f_1, \ldots, f_n$ be the function symbols used in $t$. Given a hierarchical tree $\htree = (\htreeNodes, \htreeAnnot, \htreeRoot, \htreeEdges)$ that is accepted by $\hfta$, we have:
\begin{enumerate}
\item $\htreeAnnot(v_{f_i})$ is a program that conforms to grammar $\grammars(\occurs^{-1}(f_i))$
\item The run of $\htreeAnnot(\htreeRoot)$ on $\hftaAnnot(\hftaRoot)$ is rooted with final state $q^{\denot{t}}_{s}$.
\end{enumerate}
\end{lemma}

\begin{proof}
Prove by structural induction on $t$.
\begin{itemize}
\item Base case: $t = c$.
According to the \textsf{Const} rule in Figure~\ref{fig:build-hfta}, the HFTA $\hfta$ from $t$ contains only a single node $v_c$, so the root node is also $v_c$.
The FTA $\hftaAnnot(v_c)$ is
\[
    \fta = (\bigset{q^{\denot{c}}_c}, \set{c}, \bigset{q^{\denot{c}}_c}, \bigset{c \to q^{\denot{c}}_c})
\]
Given $\htree = (\htreeNodes, \htreeAnnot, \htreeRoot, \htreeEdges)$ accepted by $\hfta$, we know tree $\htreeAnnot(v_c)$ is accepted by FTA $\hftaAnnot(v_c)$ by Definition~\ref{def:hfta-accept}. The run of $\htreeAnnot(v_c)$ is a single node $q^{\denot{c}}_c$.

\item Inductive case: $t = f(t_1, \ldots, t_m)$.
Suppose $\grammars, \occurs \vdash t_i \leadsto \hfta_i = (\hftaNodes_i, \hftaAnnot_i, v_{r_i}, \hftaTrans_i)$ for $i \in [1, m]$.
Given a hierarchical tree $\htree = (\htreeNodes, \htreeAnnot, \htreeRoot, \htreeEdges)$ accepted by $\hfta$, we denote its $t_i$ part as $\htree_i = (\htreeNodes_i, \htreeAnnot_i, v_{r_i}, \htreeEdges_i)$.
On one hand, $\htreeNodes = \cup^{m}_{i=1} \htreeNodes_i \cup \set{v_f}$. $\htreeEdges = \cup^{m}_{i=1} \htreeEdges_i \cup \set{(v_f, v_{r_i}) ~|~ i \in [1, m]}$. $\htreeAnnot_i(v) = \htreeAnnot(v)$ for all $v \in V_i$.
On the other hand, we have $\hftaAnnot_i(v) = \hftaAnnot(v)$ for all $v \in V_i$ based on the \textsf{Func} rule in Figure~\ref{fig:build-hfta}.
Also, $\hftaTrans = \cup^{m}_{i=1} \hftaTrans_i \cup \bigset{q^c_{s_i} \to q^c_{x_i} ~|~ q^c_{s_i} \in Q_{f_i}, i \in [1, m]}$ where $Q_{f_i}$ is the final state set of $\hftaAnnot(v_{r_i})$.

\hspace{0.1in} According to the acceptance condition of Definition~\ref{def:hfta-accept}, we know $\htree_i$ is accepted by $\hfta_i$ for all $i \in [1, m]$.
By inductive hypothesis, we have $\htreeAnnot_i(v_{f_i})$ conforms to grammar $\grammars(\occurs^{-1}(f_i))$ for all $v_{f_i} \in V_i$.
Moreover, we have proved $\htreeAnnot(v_f)$ conforms to grammar $\grammars(\occurs^{-1}(f))$ in Lemma~\ref{lem:sound-fta}.
Therefore, any function symbol $f_i$ in $t$ conforms to its grammar $\grammars(\occurs^{-1}(f_i))$.
Again, by inductive hypothesis, we have the run of $\htreeAnnot_i(v_{r_i})$ on $\hftaAnnot_i(v_{r_i})$ is rooted with final state $q^{\denot{t_i}}_{s_i}$ for $i \in [1, m]$.
Based on Lemma~\ref{lem:sound-fta}, the run of $\htreeAnnot(v_f)$ on $\hftaAnnot(v_f)$ is rooted with final state $q^c_s$ where $c = \denot{t(\denot{t_i}, \ldots, \denot{t_m})} = \denot{t}$.
Hence, the lemma is proved for $t = f(t_1, \ldots, t_m)$.
\end{itemize}
By principle of structural induction, we have proved the lemma.
\end{proof}

\begin{lemma}\label{lem:sound-relaxed-formula}
Suppose $\grammars,\occurs \vdash \Phi \leadsto \hfta$ according to the rules from Figure~\ref{fig:build-hfta}, where $\Phi$ is a ground formula, and let $f_1, \ldots, f_n$ be the function symbols used in $\Phi$. Given a hierarchical tree $\htree = (\htreeNodes, \htreeAnnot, \htreeRoot, \htreeEdges)$ that is accepted by $\hfta$, we have:
\begin{enumerate}
\item $\htreeAnnot(v_{f_i})$ is a program that conforms to grammar $\grammars(\occurs^{-1}(f_i))$
\item If the run of $\htreeAnnot(\htreeRoot)$ on $\hftaAnnot(\hftaRoot)$ is rooted with final state $q^\top_{s_0}$, then $\htreeAnnot(v_{f_1}), \ldots, \htreeAnnot(v_{f_n}) \models \Phi$
\item If the run of $\htreeAnnot(\htreeRoot)$ on $\hftaAnnot(\hftaRoot)$ is rooted with final state $q^\bot_{s_0}$, then $\htreeAnnot(v_{f_1}), \ldots, \htreeAnnot(v_{f_n}) \not\models \Phi$
\end{enumerate}
\end{lemma}

\begin{proof}
Prove by structural induction on $\Phi$.
\begin{itemize}
\item Base case: $\Phi = t_1 ~ op ~ t_2$.
Suppose $\grammars, \occurs \vdash t_1 \leadsto \hfta_1 = (\hftaNodes_1, \hftaAnnot_1, v_{r_1}, \hftaTrans_1)$ and $\grammars, \occurs \vdash t_2 \leadsto \hfta_2 = (\hftaNodes_2, \hftaAnnot_2, v_{r_2}, \hftaTrans_2)$.
Given a hierarchical tree $\htree = (\htreeNodes, \htreeAnnot, \htreeRoot, \htreeEdges)$ accepted by $\hfta$, we denote its $t_1$ part as $\htree_1 = (\htreeNodes_1, \htreeAnnot_1, v_{r_1}, \htreeEdges_1)$ and denote the $t_2$ part as $\htree_2 = (\htreeNodes_2, \htreeAnnot_2, v_{r_2}, \htreeEdges_2)$.
On one hand, $\htreeNodes = \htreeNodes_1 \cup \htreeNodes_2 \cup \set{v_{op}}$. $\htreeEdges = \htreeEdges_1 \cup \htreeEdges_2 \cup \set{(v_{op}, v_{r_1}), (v_{op}, v_{r_2})}$. $\htreeAnnot_i(v) = \htreeAnnot(v)$ for all $v \in V_i$ where $i = 1, 2$.
On the other hand, we have $\hftaAnnot_i(v) = \hftaAnnot(v)$ for all $v \in V_i$ where $i = 1, 2$ based on the \textsf{Logical} rule in Figure~\ref{fig:build-hfta}.
Also, $\hftaTrans = \hftaTrans_1 \cup \hftaTrans_2 \cup \bigset{q^c_{s_i} \to q^c_{x_i} ~|~ q^c_{s_i} \in Q_{f_i}, i = 1, 2}$ where $Q_{f_i}$ is the final state set of $\hftaAnnot(v_{r_i})$.

\hspace{0.1in} According to the acceptance condition of Definition~\ref{def:hfta-accept}, we know $\htree_1$ is accepted by $\hfta_1$ and $\htree_2$ is accepted by $\hfta_2$.
By Lemma~\ref{lem:sound-relaxed-term}, we have $\htreeAnnot_i(v_f)$ conforms to grammar $\grammars(\occurs^{-1}(f))$ for all $v_f \in V_i$ where $i = 1, 2$.
Therefore, any function symbol $f_i$ in $\Phi$ conforms to its grammar $\grammars(\occurs^{-1}(f_i))$.
Also by Lemma~\ref{lem:sound-relaxed-term}, we have the run of $\htreeAnnot_i(v_{r_i})$ on $\hftaAnnot_i(v_{r_i})$ is rooted with final state $q^{\denot{t_i}}_{s_i}$ for $i = 1, 2$.
In addition, consider that \textsf{Logical} rule will add transitions $\bigset{op(q^{c_1}_{x_1}, q^{c_2}_{x_2}) \to q^c_{s_0} ~|~ c = \denot{op}(c_1, c_2)}$ to $\hftaAnnot(\hftaRoot)$, the run will reach state $q^\top_{s_0}$ if $t_1 ~ op ~ t_2$ evaluates to \emph{true} and reach state $q^\bot_{s_0}$ otherwise.
Hence, the lemma is proved for $\Phi = t_1 ~ op ~ t_2$.

\item Inductive case: $\Phi = \Phi_1 \odot \Phi_2 ~ (\odot = \land, \lor).$
Suppose $\grammars, \occurs \vdash \Phi_1 \leadsto \hfta_1 = (\hftaNodes_1, \hftaAnnot_1, v_{r_1}, \hftaTrans_1)$ and $\grammars, \occurs \vdash \Phi_2 \leadsto \hfta_2 = (\hftaNodes_2, \hftaAnnot_2, v_{r_2}, \hftaTrans_2)$.
Given a hierarchical tree $\htree = (\htreeNodes, \htreeAnnot, \htreeRoot, \htreeEdges)$ accepted by $\hfta$, we denote its $\Phi_1$ part as $\htree_1 = (\htreeNodes_1, \htreeAnnot_1, v_{r_1}, \htreeEdges_1)$ and denote the $\Phi_2$ part as $\htree_2 = (\htreeNodes_2, \htreeAnnot_2, v_{r_2}, \htreeEdges_2)$.
On one hand, $\htreeNodes = \htreeNodes_1 \cup \htreeNodes_2 \cup \set{v_{\odot}}$. $\htreeEdges = \htreeEdges_1 \cup \htreeEdges_2 \cup \set{(v_{\odot}, v_{r_1}), (v_{\odot}, v_{r_2})}$. $\htreeAnnot_i(v) = \htreeAnnot(v)$ for all $v \in V_i$ where $i = 1, 2$.
On the other hand, we have $\hftaAnnot_i(v) = \hftaAnnot(v)$ for all $v \in V_i$ where $i = 1, 2$ based on the \textsf{Logical} rule in Figure~\ref{fig:build-hfta}.
Also, $\hftaTrans = \hftaTrans_1 \cup \hftaTrans_2 \cup \bigset{q^c_{s_i} \to q^c_{x_i} ~|~ q^c_{s_i} \in Q_{f_i}, i = 1, 2}$ where $Q_{f_i}$ is the final state set of $\hftaAnnot(v_{r_i})$.

\hspace{0.1in} According to the acceptance condition of Definition~\ref{def:hfta-accept}, we know $\htree_1$ is accepted by $\hfta_1$ and $\htree_2$ is accepted by $\hfta_2$.
By inductive hypothesis, we have $\htreeAnnot_i(v_f)$ conforms to grammar $\grammars(\occurs^{-1}(f))$ for all $v_f \in V_i$ where $i = 1, 2$.
Therefore, any function symbol $f_i$ in $\Phi$ conforms to its grammar $\grammars(\occurs^{-1}(f_i))$.
Also by inductive hypothesis, we know that if the run of $\htreeAnnot_i(v_{r_i})$ on $\hftaAnnot_i(v_{r_i})$ is rooted with final state $q^\top_{s_i}$ then $\htreeAnnot_i(v_{f_i}), v_{f_i} \in V_i $ collectively satisfy $\Phi_i$ for $i = 1, 2$.
In addition, the \textsf{Logical} rule is essentially building an FTA that accepts programs with boolean input value $\denot{\Phi_1}$ and $\denot{\Phi_2}$. The run is rooted with $q^\top_{s_0}$ if $\denot{\Phi_1} \odot \denot{\Phi_2}$ evaluates to \emph{true}; otherwise, the run is rooted with $q^\bot_{s_0}$.
Hence, the lemma is proved for $\Phi = \Phi_1 \odot \Phi_2$.

\item Inductive case: $\Phi = \neg \Phi_1$.
Suppose $\grammars, \occurs \vdash \Phi_1 \leadsto \hfta_1 = (\hftaNodes_1, \hftaAnnot_1, v_{r_1}, \hftaTrans_1)$.
Given a hierarchical tree $\htree = (\htreeNodes, \htreeAnnot, \htreeRoot, \htreeEdges)$ accepted by $\hfta$, we denote its $\Phi_1$ part as $\htree_1 = (\htreeNodes_1, \htreeAnnot_1, v_{r_1}, \htreeEdges_1)$.
On one hand, $\htreeNodes = \htreeNodes_1 \cup \set{v_{\neg}}$. $\htreeEdges = \htreeEdges_1 \cup \set{(v_{\neg}, v_{r_1})}$. $\htreeAnnot_1(v) = \htreeAnnot(v)$ for all $v \in V_1$.
On the other hand, we have $\hftaAnnot_1(v) = \hftaAnnot(v)$ for all $v \in V_1$ based on the \textsf{Neg} rule in Figure~\ref{fig:build-hfta}.
Also, $\hftaTrans = \hftaTrans_1 \cup \bigset{q^c_s \to q^c_{x_1} ~|~ q^c_s \in Q_{f_1}}$ where $Q_{f_1}$ is the final state set of $\hftaAnnot(v_{r_1})$.

\hspace{0.1in} According to the acceptance condition of Definition~\ref{def:hfta-accept}, we know $\htree_1$ is accepted by $\hfta_1$.
By inductive hypothesis, we have $\htreeAnnot_1(v_f)$ conforms to grammar $\grammars(\occurs^{-1}(f))$ for all $v_f \in V_1$.
Therefore, any function symbol $f_i$ in $\Phi$ conforms to its grammar $\grammars(\occurs^{-1}(f_i))$.
Also by inductive hypothesis, we have the run of $\htreeAnnot_1(v_{r_1})$ on $\hftaAnnot_1(v_{r_1})$ is rooted with final state $q^{\denot{t_1}}_{s_1}$.
In addition, the \textsf{Neg} rule is building an FTA that accepts programs with boolean input value $\denot{\Phi_1}$. The run is rooted with $q^\top_{s_0}$ if $\neg \denot{\Phi_1}$ evaluates to \emph{true}; otherwise, the run is rooted with $q^\bot_{s_0}$.
Hence, the lemma is proved for $\Phi = \neg \Phi_1$.
\end{itemize}
By principle of structural induction, we have proved the lemma.
\end{proof}

\begin{proof}[Proof of Theorem~\ref{thm:sound-relaxed}]
Suppose $\grammars, \occurs \vdash \Phi \twoheadrightarrow \hfta$ and $\hfta = (\hftaNodes, \hftaAnnot, \hftaRoot, \hftaTrans)$, we know $\grammars, \occurs \vdash \Phi \leadsto \hfta'$ with $\hfta' = (\hftaNodes, \hftaAnnot', \hftaRoot, \hftaTrans)$ according to the \textrm{Final} rule in Figure~\ref{fig:build-hfta}.
The only difference between $\hftaAnnot$ and $\hftaAnnot'$ is that $\hftaAnnot(\hftaRoot) = (\ftaStates, \ftaAlphabet, \set{q^\top_{s_0}}, \ftaTrans)$ but $\hftaAnnot'(\hftaRoot) = (\ftaStates, \ftaAlphabet, \set{q^\top_{s_0}, q^\bot_{s_0}}, \ftaTrans)$.
By Definition~\ref{def:hfta-accept}, we know $\lang(\hfta) \subseteq \lang(\hfta')$. Thus, if hierarchical tree $\htree$ is accepted by $\hfta$, then $\htree$ is also accepted by $\hfta'$.
According to Lemma~\ref{lem:sound-relaxed-formula}, we have $\htreeAnnot(v_{f_i})$ conforms to grammar $\grammars(\occurs^{-1}(f_i))$. In addition, we have $\htreeAnnot(v_{f_1}), \ldots, \htreeAnnot(v_{f_n}) \models \Phi$ if the run of $\htreeAnnot(\htreeRoot)$ is rooted with final state $q^\top_{s_0}$.
Since $\hfta$ only has one final state $q^\top_{s_0}$ in $\hftaAnnot(\hftaRoot)$, we know the accepting run of $\htreeAnnot(\htreeRoot)$ on $\hftaAnnot(\hftaRoot)$ must be rooted with $q^\top_{s_0}$. Thus, $\htreeAnnot(v_{f_1}), \ldots, \htreeAnnot(v_{f_n}) \models \Phi$.
\end{proof}


\begin{lemma}\label{lem:complete-fta}
Suppose FTA $\fta = (\ftaStates, \ftaAlphabet, \ftaFinal, \ftaTrans)$ is built with grammar $G$ and initial states $[\ftaStates_1, \ldots, \ftaStates_m]$ using rules from Figure~\ref{fig:buildFTA}. If there is a program $P$ that conforms to grammar $G$ such that $\denot{P}(a_1, \ldots, a_m) = a$ where $q^{a_i}_{s_i} \in Q_i$ for $i \in [1, m]$, then there exists a tree accepted by $\fta$ and the run is rooted with final state $q^a_s$ from initial states $q^{a_1}_{x_1}, \ldots, q^{a_m}_{x_m}$.
\end{lemma}

\begin{proof}
To prove by induction on height of $P$'s AST, we first strengthen the lemma to: the run of tree $t$ is rooted with state $q^a_s$ from initial states $q^{a_1}_{x_1}, \ldots, q^{a_m}_{x_m}$ where $a = \denot{t}(a_1, \ldots, a_m)$.
\begin{itemize}
\item Base case: height is 1.
In this case, $P$ only has one variable $x$, and $q^c_{s}$ is passed as an input.
Consider the \textsf{Input} rule in Figure~\ref{fig:buildFTA}, we know $q^c_{x} \in \ftaStates$ and $x \to q^c_x \in \ftaTrans$.
Then $x$ is accepted by $\fta$, and the run is rooted with final state $q^c_{x}$.

\item Inductive case: Suppose the lemma is correct for AST with height no larger than $h$, and the goal is to prove it is correct for AST of height $h+1$.
Without loss of generality, assume $P = \sigma(P_1, \ldots, P_n)$ is of height $h+1$, so the height of partial programs $P_1, \ldots, P_n$ is at most $h$.
By inductive hypothesis, there exists trees $t_1, \ldots, t_n$ whose runs are rooted with state $q^{c_i}_{s_i}$ form initial states $q^{a_1}_{x_1}, \ldots, q^{a_m}_{x_m}$, and $\denot{t_i} = c_i$ for all $i \in [1, n]$.
According to the \textsf{Prod} rule in Figure~\ref{fig:buildFTA}, we know the run of $t = \sigma(t_1, \ldots, t_n)$ is rooted with state $\denot{\sigma(\denot{t_1}, \ldots, \denot{t_n})} = \denot{\sigma(c_1, \ldots, c_n)}$. Hence, the inductive case is proved.
\end{itemize}
Consider $P$ is a program (not partial program), if the run of tree $t$ is rooted with state $q^a_s$, then $s$ must be a starting symbol.
By the \textsf{Output} rule in Figure~\ref{fig:buildFTA}, $q^a_s$ is a final state. Therefore, tree $t$ is accepted by $\fta$, and we have proved the lemma.
\end{proof}

\begin{lemma}\label{lem:complete-relaxed-term}
Let $t$ be a term where every function symbol $f_1, \ldots, f_n$ occurs exactly once, and suppose we have $\grammars,\occurs \vdash t \leadsto \hfta$. If there are implementations $P_i$ of $f_i$ such that $\denot{t} = a$, where $P_i$ conforms to grammar  $\grammars(\occurs^{-1}(f_i))$, then there exists a hierarchical tree $\htree = (\htreeNodes, \htreeAnnot, \htreeRoot, \htreeEdges)$ accepted by $\hfta$ such that $\htreeAnnot(v_{f_i}) = P_i$ for all $i \in [1,n]$, and the run of $\htreeAnnot(\htreeRoot)$ on $\hftaAnnot(\hftaRoot)$ is rooted with final state $q^a_s$.
\end{lemma}

\begin{proof}
Prove by structural induction on $t$.
\begin{itemize}
\item Base case: $t = c$.
Suppose the HFTA built from \textsf{Const} rule in Figure~\ref{fig:build-hfta} is $\hfta = (\set{v_c}, \hftaAnnot, v_c, \hftaTrans)$ and $\hftaAnnot(v_c)$ is
\[
    \fta = (\bigset{q^{\denot{c}}_c}, \set{c}, \bigset{q^{\denot{c}}_c}, \bigset{c \to q^{\denot{c}}_c})
\]
Consider the hierarchical tree $\htree = (\set{v_c}, \htreeAnnot, v_c, \emptyset)$ where $\htreeAnnot(v_c)$ is a tree that only contains a single node $c$.
By Definition~\ref{def:hfta-accept}, $\htree$ is accepted by $\hfta$ and the run of $\htreeAnnot(v_c)$ on $\hftaAnnot(v_c)$ is rooted with final state $q^{\denot{c}}_c$.

\item Inductive case: $t = f(t_1, \ldots, t_m)$.
Given $\grammars, \occurs \vdash t \leadsto \hfta = (\hftaNodes, \hftaAnnot, v_f, \hftaTrans)$, we denote its $t_i$ part as $\hfta_i = (\hftaNodes_i, \hftaAnnot_i, v_{r_i}, \hftaTrans_i)$ for $i \in [1, m]$. 
In addition, we have $\hftaNodes_i(v_i) = \hftaNodes(v_i)$ and $\hftaAnnot_i(v_i) = \hftaNodes(v_i)$ for all node $v_i \in \hftaNodes_i$. $\hftaTrans_i \subset \hftaTrans$.
Note that the $t_i$ part is recognized by inter-FTA transition $q^c_s \to q^c_{x_i} \in \hftaTrans$ where $q^c_s$ is a state in $\hftaAnnot_i(v_{r_i})$ and $q^c_{x_i}$ is a state in $\hftaAnnot(v_f)$.
According to the \textsf{Func} rule in Figure~\ref{fig:build-hfta}, we have $\grammars, \occurs \vdash t_i \leadsto \hfta_i$ for $i \in [1, m]$.
By inductive hypothesis, we know there exists a hierarchical tree $\htree_i = (\htreeNodes_i, \htreeAnnot_i, v_{r_i}, \htreeEdges_i)$ accepted by $\hfta_i$ and the run of $\htreeAnnot_i(v_{r_i})$ on $\hftaAnnot_i(v_{r_i})$ is rooted with final state $q^{\denot{t_i}}_{s_i}$ for all $i \in [1, m]$.
Moreover, by Lemma~\ref{lem:complete-fta}, we know there exists a tree $T_f$ accepted by $\hftaAnnot(v_f)$ and the run is rooted with final state $q^a_s$ from initial states $q^{\denot{t_1}}_{s_1}, \ldots, q^{\denot{t_m}}_{s_m}$, where $a = \denot{f(\denot{t_1}, \ldots, \denot{t_m})}$.
Now consider the hierarchical tree
\[
    \htree = (\cup^m_{i=1}\htreeNodes_i \cup \set{v_f}, ~ \cup^m_{i=1}\htreeAnnot_i \cup \set{v_f \mapsto T_f}, ~ v_f, ~ \cup^m_{i=1} (\htreeEdges_i \cup \set{(v_f, v_{r_i})} ))
\]
By acceptance condition of HFTA in Defintion~\ref{def:hfta-accept}, $\htree$ is accepted by $\hfta$.
Thus, the lemma is proved for $t = f(t_1, \ldots, t_m)$ case.
\end{itemize}
By principle of structural induction, we have proved the lemma.
\end{proof}

\begin{lemma}\label{lem:complete-relaxed-formula}
Let $\Phi$ be a ground formula where every function symbol $f_1, \ldots, f_n$ occurs exactly once, and suppose we have $\grammars,\occurs \vdash \Phi \leadsto \hfta$. If there are implementations $P_i$ of $f_i$ where $P_i$ conforms to grammar $\grammars(\occurs^{-1}(f_i))$, then there exists a hierarchical tree $\htree = (\htreeNodes, \htreeAnnot, \htreeRoot, \htreeEdges)$ accepted by $\hfta$ such that $\htreeAnnot(v_{f_i}) = P_i$ for all $i \in [1,n]$, and
\begin{enumerate}
\item If $P_1, \ldots, P_n \models \Phi$, the run of $\htreeAnnot(\htreeRoot)$ on $\hftaAnnot(\hftaRoot)$ is rooted with final state $q^\top_{s_0}$.
\item If $P_1, \ldots, P_n \not\models \Phi$, the run of $\htreeAnnot(\htreeRoot)$ on $\hftaAnnot(\hftaRoot)$ is rooted with final state $q^\bot_{s_0}$.
\end{enumerate}
\end{lemma}

\begin{proof}
Prove by structural induction on $\Phi$.
\begin{itemize}
\item Base case: $\Phi = t_1 ~ op ~ t_2$.
Given $\grammars, \occurs \vdash t \leadsto \hfta = (\hftaNodes, \hftaAnnot, v_{op}, \hftaTrans)$, we denote its $t_i$ part as $\hfta_i = (\hftaNodes_i, \hftaAnnot_i, v_{r_i}, \hftaTrans_i)$ for $i = 1, 2$. 
In addition, we have $\hftaNodes_i(v_i) = \hftaNodes(v_i)$ and $\hftaAnnot_i(v_i) = \hftaNodes(v_i)$ for all node $v_i \in \hftaNodes_i$. $\hftaTrans_i \subset \hftaTrans$.
Note that the $t_i$ part is recognized by inter-FTA transition $q^c_s \to q^c_{x_i} \in \hftaTrans$ where $q^c_s$ is a state in $\hftaAnnot_i(v_{r_i})$ and $q^c_{x_i}$ is a state in $\hftaAnnot(v_{op})$.
According to the \textsf{Logical} rule in Figure~\ref{fig:build-hfta}, we have $\grammars, \occurs \vdash t_i \leadsto \hfta_i$ for $i = 1, 2$.
By Lemma~\ref{lem:complete-relaxed-term}, we know there exists a hierarchical tree $\htree_i = (\htreeNodes_i, \htreeAnnot_i, v_{r_i}, \htreeEdges_i)$ accepted by $\hfta_i$ and the run of $\htreeAnnot_i(v_{r_i})$ on $\hftaAnnot_i(v_{r_i})$ is rooted with final state $q^{\denot{t_i}}_{s_i}$ for $i = 1, 2$.
Moreover, based on the \textsf{Logical} rule, we know there exists a tree $T_{op} = op(x_1, x_2)$ accepted by $\hftaAnnot(v_{op})$. The run is rooted with $q^\top_{s_0}$ if $\denot{t_1} ~ op ~ \denot{t_2}$ evaluates to \emph{true} (i.e., $P_1, \ldots, P_n \models t_1 ~ op t_2$); the run is rooted with $q^\bot_{s_0}$ otherwise.
Now consider the hierarchical tree
\[
    \htree = (\htreeNodes_1 \cup \htreeNodes_2 \cup \set{v_{op}}, ~ \htreeAnnot_1 \cup \htreeAnnot_2 \cup \set{v_{op} \mapsto T_{op}}, ~ v_{op}, ~ \htreeEdges_1 \cup \htreeEdges_2 \cup \set{(v_{op}, v_{r_1}), (v_{op}, v_{r_2})} ))
\]
By acceptance condition of HFTA in Defintion~\ref{def:hfta-accept}, $\htree$ is accepted by $\hfta$.
Thus, the lemma is proved for $\Phi = t_1 ~ op ~ t_2$.

\item Inductive case: $\Phi = \Phi_1 \odot \Phi_2 ~ (\odot = \land, \lor)$.
Given $\grammars, \occurs \vdash t \leadsto \hfta = (\hftaNodes, \hftaAnnot, v_{\odot}, \hftaTrans)$, we denote its $\Phi_i$ part as $\hfta_i = (\hftaNodes_i, \hftaAnnot_i, v_{r_i}, \hftaTrans_i)$ for $i = 1, 2$. 
In addition, we have $\hftaNodes_i(v_i) = \hftaNodes(v_i)$ and $\hftaAnnot_i(v_i) = \hftaNodes(v_i)$ for all node $v_i \in \hftaNodes_i$. $\hftaTrans_i \subset \hftaTrans$.
Note that the $\Phi_i$ part is recognized by inter-FTA transition $q^c_s \to q^c_{x_i} \in \hftaTrans$ where $q^c_s$ is a state in $\hftaAnnot_i(v_{r_i})$ and $q^c_{x_i}$ is a state in $\hftaAnnot(v_{\odot})$.
According to the \textsf{Logical} rule in Figure~\ref{fig:build-hfta}, we have $\grammars, \occurs \vdash \Phi_i \leadsto \hfta_i$ for $i = 1, 2$.
By inductive hypothesis, we know there exists a hierarchical tree $\htree_i = (\htreeNodes_i, \htreeAnnot_i, v_{r_i}, \htreeEdges_i)$ accepted by $\hfta_i$ and the run of $\htreeAnnot_i(v_{r_i})$ on $\hftaAnnot_i(v_{r_i})$ is rooted with final state $q^\top_{s_0}$ if $\Phi_i$ is satisfied; the run is rooted with final state $q^\top_{s_0}$ otherwise.
Moreover, based on the \textsf{Logical} rule, we know there exists a tree $T_{\odot} = \odot(x_1, x_2)$ accepted by $\hftaAnnot(v_{\odot})$. The run is rooted with $q^\top_{s_0}$ if $\denot{\Phi_1} \odot \denot{\Phi_2}$ evaluates to \emph{true}, but rooted with $q^\bot_{s_0}$ otherwise.
Now consider the hierarchical tree
\[
    \htree = (\htreeNodes_1 \cup \htreeNodes_2 \cup \set{v_{\odot}}, ~ \htreeAnnot_1 \cup \htreeAnnot_2 \cup \set{v_{\odot} \mapsto T_{\odot}}, ~ v_{\odot}, ~ \htreeEdges_1 \cup \htreeEdges_2 \cup \set{(v_{\odot}, v_{r_1}), (v_{\odot}, v_{r_2})} )
\]
By acceptance condition of HFTA in Defintion~\ref{def:hfta-accept}, $\htree$ is accepted by $\hfta$.
Thus, the lemma is proved for $\Phi = \Phi_1 \odot \Phi_2$.

\item Inductive case: $\Phi = \neg \Phi_1$.
Given $\grammars, \occurs \vdash t \leadsto \hfta = (\hftaNodes, \hftaAnnot, v_{\neg}, \hftaTrans)$, we denote its $\Phi_1$ part as $\hfta_1 = (\hftaNodes_1, \hftaAnnot_1, v_{r_1}, \hftaTrans_1)$.
In addition, we have $\hftaNodes_1(v_1) = \hftaNodes(v_1)$, $\hftaAnnot_1(v_1) = \hftaNodes(v_1)$, and $\hftaTrans_1 \subset \hftaTrans$.
Note that the $\Phi_1$ part is recognized by inter-FTA transition $q^c_s \to q^c_{x_1} \in \hftaTrans$ where $q^c_s$ is a state in $\hftaAnnot_1(v_{r_1})$ and $q^c_{x_1}$ is a state in $\hftaAnnot(v_{\neg})$.
According to the \textsf{Neg} rule in Figure~\ref{fig:build-hfta}, we have $\grammars, \occurs \vdash \Phi_1 \leadsto \hfta_1$.
By inductive hypothesis, we know there exists a hierarchical tree $\htree_1 = (\htreeNodes_1, \htreeAnnot_1, v_{r_1}, \htreeEdges_1)$ accepted by $\hfta_1$ and the run of $\htreeAnnot_1(v_{r_1})$ on $\hftaAnnot_1(v_{r_1})$ is rooted with final state $q^\top_{s_0}$ if $\Phi_1$ is satisfied; the run is rooted with final state $q^\top_{s_0}$ otherwise.
Moreover, based on the \textsf{Neg} rule, we know there exists a tree $T_{\neg} = \neg(x_1)$ accepted by $\hftaAnnot(v_{\neg})$. The run is rooted with $q^\top_{s_0}$ if $\neg \denot{\Phi_1}$ evaluates to \emph{true}, but rooted with $q^\bot_{s_0}$ otherwise.
Now consider the hierarchical tree
\[
    \htree = (\htreeNodes_1 \cup \set{v_{\neg}}, ~ \htreeAnnot_1 \cup \set{v_{\neg} \mapsto T_{\neg}}, ~ v_{\neg}, ~ \htreeEdges_1 \cup \set{(v_{\neg}, v_{r_1})} )
\]
By acceptance condition of HFTA in Defintion~\ref{def:hfta-accept}, $\htree$ is accepted by $\hfta$.
Thus, the lemma is proved for $\Phi = \neg \Phi_1$.
\end{itemize}
By principle of structural induction, we have proved the lemma.
\end{proof}

\begin{proof}[Proof of Theorem~\ref{thm:complete-relaxed}]
Suppose $\grammars, \occurs \vdash \Phi \twoheadrightarrow \hfta$ and $\hfta = (\hftaNodes, \hftaAnnot, \hftaRoot, \hftaTrans)$, we have $\grammars, \occurs \vdash \Phi \leadsto \hfta'$ with $\hfta' = (\hftaNodes, \hftaAnnot', \hftaRoot, \hftaTrans)$ according to the \textrm{Final} rule in Figure~\ref{fig:build-hfta}.
The only difference between $\hftaAnnot$ and $\hftaAnnot'$ is that $\hftaAnnot(\hftaRoot) = (\ftaStates, \ftaAlphabet, \set{q^\top_{s_0}}, \ftaTrans)$ but $\hftaAnnot'(\hftaRoot) = (\ftaStates, \ftaAlphabet, \set{q^\top_{s_0}, q^\bot_{s_0}}, \ftaTrans)$.
By Lemma~\ref{lem:complete-relaxed-formula}, we know there exists a hierarchical tree $\htree = (\htreeNodes, \htreeAnnot, \htreeRoot, \htreeEdges)$ accepted by $\hfta'$ such that $\htreeAnnot(v_{f_i}) = P_i$ for all $i \in [1, n]$. Furthermore, since $P_1, \ldots, P_n \models \Phi$, we know the run of $\htreeAnnot(\htreeRoot)$ on $\hftaAnnot(\hftaRoot)$ is rooted with final state $q^\top_{s_0}$.
Since $q^\top_{s_0}$ is also a final state in $\hftaAnnot'(\hftaRoot)$, the run of $\htreeAnnot(\htreeRoot)$ on FTA $\hftaAnnot(\hftaRoot)$ is the same as the run of $\htreeAnnot(\htreeRoot)$ on FTA $\hftaAnnot'(\hftaRoot)$.
According to the HFTA accepting condition in Definition~\ref{def:hfta-accept}, $\htree$ is also accepted by $\hfta$.
\end{proof}


\begin{lemma}\label{lem:find-sound}
Suppose we have a ground relational specification $\Phi$, its relaxation $\Phi'$ with symbol mapping $\occurs$, and an HFTA $\hfta$ built from $\Phi'$, $\occurs$ and grammars $\grammars$. If the \textsc{FindProgs} procedure returns programs $\programs$ (not \emph{null}) given input $\hfta$ and $\occurs$, then $\programs$ satisfy $\Phi$, and program $P_i$ for function $f_i$ conforms to the grammar $\grammars(f_i)$.
\end{lemma}

\begin{proof}
If \textsc{FindProgs} actually returns programs $\programs$, then there must be a hierarchical tree $\htree = (\htreeNodes, \htreeAnnot, \htreeRoot, \htreeEdges)$ accepted by $\hfta$. Otherwise, it will skip the for loop and directly return \emph{null}.
For ease of illustration, we assign $\hfta_0 = \hfta$ and assign $\htree_0 = \htree = (\htreeNodes, \htreeAnnot_0, \htreeRoot, \htreeEdges)$.
Observe that the \textsc{FindProgs} procedure assigns one program to a function symbol at a time, and terminates when all function symbols are successfully assigned.
Without loss of generality, we assume the function symbols are assigned in the order of $f_1, \ldots, f_n$.
We denote by $\htree_i = (\htreeNodes, \htreeAnnot_i, \htreeRoot, \htreeEdges)$ the hierarchical tree and $\htreeAnnot_i$ is obtained by substituting \emph{all} $\htreeAnnot_{i-1}(v_{f'_i})$ in $\htree_{i-1}$ with program tree $P_i$, where $f'_i$ is a symbol in $\occurs(f_i)$. We also denote the HFTA of invoking \textsc{Propagate}($\hfta_{i-1}, P_i, f_i, \occurs$) by $\hfta_i$.

To prove programs $P_1, \ldots, P_n$ in $\programs$ satisfy $\Phi$ and program $P_i$ for function $f_i$ conforms to grammar $\grammars(f_i)$, consider the following invariant: For any integer $i \in [0, n]$,
\begin{enumerate}
\item $\lang(\hfta_i) \subseteq \lang(\hfta)$.
\item For all $k \leq i$, function symbol $f_k$ in original ground specification $\Phi$ is consistently assigned to program $P_k$, i.e. $P_k = \htreeAnnot_i(v_{f'_k})$ for all $f'_k \in \occurs(f_k)$.
\end{enumerate}
First, let us prove it is indeed an invariant.
\begin{itemize}
\item Base case: $i = 0$. Given $\hfta_0 = \hfta$, we have $\lang(\hfta_0) \subseteq \lang(\hfta)$. (2) trivially holds because no program has been assigned for $i = 0$.
\item Inductive case: Suppose (1)(2) hold for $i = j$, the goal is to prove they also hold for $i = j+1$.

\begin{enumerate}
\item Since \textsc{Propagate} will only remove final states and inter-FTA transitions from $\hfta_j$, we know $\lang(\hfta_{j+1}) \subseteq \lang(\hfta_j)$ by acceptance condition in Definition~\ref{def:hfta-accept}. Since $\lang(\hfta_j) \subseteq \lang(\hfta)$ by inductive hypothesis, we have $\lang(\hfta_{j+1}) \subseteq \lang(\hfta)$.
\item By definition of $\htree_{j+1}$, $\htreeAnnot_{j+1}(v_{f'_{j+1}}) = P_{j+1}$ holds for all $f'_{j+1} \in \occurs(f_{j+1})$. Also by inductive hypothesis, we know for all $k \leq j$, $P_k = \htreeAnnot_k(v_{f'_k})$. Thus, for all $k \leq j+1$, we have $P_k = \htreeAnnot_k(v_{f'_k})$ for all $f'_k \in \occurs(f_k)$.
\end{enumerate}
\end{itemize}

After assigning programs to all function symbols, we have
\begin{enumerate}
\item $\lang(\hfta_n) \subseteq \lang(\hfta)$.
\item For all function symbol $f_i$ in $\Phi$, function occurrences $f'_i$ are consistently assigned to program $P_i$ for all $f'_i \in \occurs(f_i)$.
\end{enumerate}

Since the programs $\programs$ are accepted by the last HFTA $\hfta_n$, i.e. $\htree_n \in \lang(\hfta_n)$ (otherwise \textsc{Propagate} would return \emph{null} due to the recurrent failure of emptiness check on $\hfta_n$, we have $\htree_n \in \lang(\hfta)$.
By Theorem~\ref{thm:sound-relaxed}, we have $\htreeAnnot(v_{f'_i})$ conforms to grammar $\grammars(\occurs^{-1}(f_i))$ and $\htreeAnnot(v_{f'_1}), \ldots, \htreeAnnot(v_{f'_n}) \models \Phi'$.
Consider the fact that $f_i = \occurs^{-1}(f'_i) \Leftrightarrow f'_i \in \occurs(f_i)$ and the only difference between $\Phi$ and $\Phi'$ is captured by $\occurs$, we have proved $P_1, \ldots, P_n \models \Phi$ and $P_i$ for function $f_i$ conforms to grammar $\grammars(f_i)$.
\end{proof}

\begin{proof}[Proof of Theorem~\ref{thm:sound}]
Observe the \textsc{Synthesize} procedure shown in Algorithm~\ref{algo:cegis} can only return a set of programs $\programs$ after it invokes the \textsf{Verify} procedure and successfully verifies $\programs$ indeed satisfy the relational specification $\Psi$. Thus, the programs $\programs$ returned by \textsc{Synthesize} are guaranteed to satisfy $\Psi$.
In addition, observe that Algorithm~\ref{algo:cegis} only has two ways to generate candidate programs that may be potentially returned after verification.
\begin{enumerate}
\item Randomly generate programs $\programs_{\emph{random}}$ that conform to context-free grammars $\grammars$ in the initialization phase.
\item Generate programs $\programs_{\emph{res}}$ by the \textsc{FindProgs} procedure. According to Lemma~\ref{lem:find-sound}, for program $P_i$ of function $f_i$ in $\programs_{\emph{res}}$, program $P_i$ conforms to the grammar $\grammars(f_i)$.
\end{enumerate}
Put two cases together, we have program $P_i$ of function symbol $f_i$ conforms to grammar $\grammars(f_i)$.
\end{proof}


\begin{lemma}\label{lem:find-complete}
Suppose we have a ground relational specification $\Phi$, its relaxation $\Phi'$ with symbol mapping $\occurs$, and an HFTA $\hfta$ built from $\Phi'$, $\occurs$, and grammars $\grammars$. If there exist a set of programs that (a) satisfy $\Phi$ and (b) can be implemented using the DSLs given by $\grammars$, then invoking \textsc{FindProgs} on $\hfta$ and $\occurs$ will eventually return programs $\programs$ satisfying $\Phi$.
\end{lemma}

\begin{proof}
Suppose there exist a set of programs that satisfy $\Phi$ and can be implemented using the DSLs given by $\grammars$, we denote the program for function $f_i$ by $P_i$.
Given that the relationship between $\Phi$ and $\Phi'$ is captured by mapping $\occurs$, and $f_i = \occurs^{-1}(f'_i) \Leftrightarrow f'_i \in \occurs(f_i)$, we know that assigning program $P_i$ for \emph{all} $f'_i$ if $f'_i \in \occurs(f_i)$ would satisfy $\Phi'$, and $P_i$ conforms to grammar $\grammars(\occurs^{-1}(f'_i)$.
By Theorem~\ref{thm:complete-relaxed}, we know there exists a hierarchical tree $\htree = (\htreeNodes, \htreeAnnot, \htreeRoot, \htreeEdges)$ accepted by $\hfta$ such that $\htreeAnnot(v_{f'_i}) = P_i$ for any $i \in [1, n]$.
Observe that the \textsc{FindProgs} procedure essentially employs standard backtrack search and may potentially enumerate all possible hierarchical trees accepted by $\hfta$, it will eventually find $\htree$ and thus return the program set $\programs$.
\end{proof}

\begin{proof}[Proof of Theorem~\ref{thm:complete}]
Suppose the ground relational specification obtained by $i$-th loop iteration of Algorithm~\ref{algo:cegis} is denoted by $\Phi_i$, and $\Phi_0 = \top$, let us first prove that $\Psi \Rightarrow \Phi_i$ for any natural number $i$ by induction on $i$.
\begin{itemize}
\item Base case: $i = 0$. $\Psi \Rightarrow \top$ holds obviously.
\item Inductive case: Suppose $\Psi \Rightarrow \Phi_{i-1}$ holds, we would like to prove $\Psi \Rightarrow \Phi_i$. Given $\Phi_i = \Phi_{i-1} \land \emph{GetRelationalCounterexample}(\programs, \Psi)$, and consider that $\Psi \Rightarrow \emph{GetRelationalCounterexample}(\programs, \Psi)$ because the counterexamples are obtained by instantiating the universal quantifiers in $\Psi$, we have $\Psi \Rightarrow \Phi_i$.
\end{itemize}
According to the principle of induction, $\Psi \Rightarrow \Phi_i$ holds for any natural number $i$.
If there exist programs $\programs$ that satisfy $\Psi$ and $\programs$ can be implemented using the DSLs given by $\grammars$, then programs $\programs$ also satisfy $\Phi_i$ for any $i$.
By Lemma~\ref{lem:find-complete}, we know the \textsc{FindProgs} procedure will always return programs satisfying $\Phi_i$.
Since the Algorithm~\ref{algo:cegis} will not terminate until it finds a set of programs that satisfy $\Psi$ or the \textsc{FindProgs} procedure fails to find any programs, and we have proved \textsc{FindProgs} will never return failure, we can conclude that Algorithm~\ref{algo:cegis} will eventually terminate with programs satisfying $\Psi$.
\end{proof}

\section{Complexity Analysis}

The complexity of our inductive synthesis algorithm is parametrized by 1) the number of functions to synthesize, 2) the size of ground relational specification, 3) the size of largest FTA in the HFTA, and 4) the bound of AST size for programs under consideration.
Specifically, the ground relational specification can be relaxed in linear time $\bigO(k)$, where $k$ is the size of ground relational specification.
Since the HFTA may contain at most $k$ FTAs, it can be constructed in $\bigO(km)$ time, where $m$ is the size of largest FTA in the HFTA. In particular, $m = \Sigma_{\delta \in \Delta} |\delta|$ where $|\delta| = w+1$ for a transition $\delta$ of the form $\sigma(q_1, \ldots, q_w) \to q$.
Observe that \textsc{FindProgs} is a backtracking search procedure, its worst case complexity is $\bigO(m^{bn})$ where $n$ is the number of functions to synthesize, and $b$ is the bound of AST size for all programs under consideration.
Therefore, the overall complexity is $\bigO(km + m^{bn})$.

\fi

\end{document}